\documentclass[11pt]{article}

\usepackage{amsmath,amssymb,amsfonts,amsthm,epsfig}
\usepackage[usenames,dvipsnames]{xcolor}
\usepackage{bm,xspace}
\usepackage{tcolorbox}
\usepackage{cancel}
\usepackage{fullpage}
\usepackage{framed}
\usepackage{liyang}
\usepackage{verbatim}
\usepackage{enumitem}
\usepackage{array}
\usepackage{multirow}
\usepackage{afterpage}
\usepackage{mathrsfs}
\usepackage{pifont} 
\usepackage{chngpage}
\usepackage[normalem]{ulem}
\usepackage{boxedminipage}
\usepackage{caption}
\usepackage{subcaption}
\usepackage{microtype}

\usepackage{forest}

\usepackage{algorithm}
\usepackage{algorithmicx}
\usepackage{algpseudocode}

\usepackage{tikz}

\usetikzlibrary{calc,through,backgrounds,decorations.pathreplacing, calligraphy,arrows.meta}
\usetikzlibrary{positioning,chains,fit,shapes}

\usepackage{tikz-cd}

\usepackage{pgfplots}
\pgfplotsset{width=8cm,compat=newest}

\def\colorful{0}

\ifnum\colorful=1

\fi
\ifnum\colorful=0

\fi

\crefname{fact}{fact}{facts}

\newcommand{\odepth}{\overline{\mathrm{Depth}}}
\newcommand{\odepthmu}{\overline{\mathrm{Depth}}^\mu}
\newcommand{\odepthmuk}{\overline{\mathrm{Depth}}^{\mu^k}}

\newcommand{\Ber}{\mathrm{Ber}}
\newcommand{\BerSum}{\mathrm{BerSum}}

\newcommand{\mcD}{\mathcal D}

\newcommand{\mcT}{\mathcal T}

\newcommand{\mcY}{\mathcal Y}

\newcommand{\Bin}{\mathrm{Bin}}

\newcommand{\dens}{\mathrm{Dens}}
\newcommand{\adv}{\mathrm{Adv}}
\newcommand{\Dens}{\mathrm{Dens}}
\newcommand{\Adv}{\mathrm{Adv}}

\DeclarePairedDelimiter\abs{|}{|}

\DeclarePairedDelimiter\lzero{\|}{\|_0}

\DeclarePairedDelimiter\paren{(}{)}
\DeclarePairedDelimiter\bracket{[}{]}
\DeclarePairedDelimiter\set{\{}{\}}

\newlist{enumprop}{enumerate}{1} 
\setlist[enumprop]{label=\arabic*.,ref=\theproposition.\arabic*}

\makeatletter
\newtheorem*{rep@theorem}{\rep@title}
\newcommand{\newreptheorem}[2]{
\newenvironment{rep#1}[1]{
 \def\rep@title{#2 \ref{##1}}
 \begin{rep@theorem}\itshape}
 {\end{rep@theorem}}}
\makeatother

\newreptheorem{theorem}{Theorem}

\newcommand{\pparagraph}[1]{\bigskip \noindent {\bf {#1}}}

\title{A strong direct sum theorem for distributional query complexity\vspace{10pt}}

\author{ 
Guy Blanc \qquad  
Caleb Koch  \qquad  
 Carmen Strassle  \qquad 
Li-Yang Tan \vspace{15pt} \\
{\sl Stanford University}
\vspace{5pt}
}

\date{\small{\today}}

\begin{document}

\maketitle

\begin{abstract} 
Consider the expected query complexity of computing the $k$-fold direct product $f^{\otimes k}$ of a function $f$ to error $\varepsilon$ with respect to a distribution $\mu^k$. One strategy is to sequentially compute each of the $k$ copies to error $\varepsilon/k$ with respect to $\mu$ and apply the union bound. We prove a {\sl strong direct sum theorem} showing that this naive strategy is essentially optimal. In particular, computing a direct product necessitates a blowup in both query complexity {\sl and} error. 

Strong direct sum theorems contrast with results that only show a blowup in query complexity or error but not both. There has been a long line of such results for distributional query complexity, dating back to (Impagliazzo, Raz, Wigderson 1994) and (Nisan, Rudich, Saks 1994), but a strong direct sum theorem had been elusive.


A key idea in our work is the first use of the Hardcore Theorem (Impagliazzo 1995) in the context of query complexity. We prove a new ``resilience lemma" that accompanies it, showing that the hardcore of $f^{\otimes k}$ is likely to remain dense under arbitrary partitions of the input space.

\end{abstract} 

\thispagestyle{empty}

\newpage 
 \setcounter{tocdepth}{2}
\tableofcontents
\thispagestyle{empty}

\newpage 

\setcounter{page}{1}

\section{Introduction}

The {\sl direct sum problem} seeks to understand the ways in which the complexity of solving~$k$ independent instances of a computational task scales with~$k$. This problem and its variants such as the {\sl XOR problem}, where one only seeks to compute the XOR of the $k$ output values, have a long history in complexity theory.  Research on them dates back to Strassen~\cite{Str73} and they have since been  studied in all major computational models including boolean circuits~\cite{Yao82,Lev85,GNW11,Imp95,IW97,OD02,IJKW08,Dru13}, communication protocols~\cite{IRW94,Sha04,KSdW07,LSS08,VW08,Kla10,She11,Jai15,JPY12,BBCR10,BRWY13,Yu22}, as well as classical~\cite{IRW94,NRS94,Sha04,KSdW07,Dru12,BDK18,BB19,BKLS23,Hoza23} and quantum query complexity~\cite{ASdW06,KSdW07,Spa08,She11,AMRR11,LR13}. 

\subsection{This work} 

We focus on classical query complexity, and specifically distributional query complexity. Distributional complexity, also known as average-case complexity, is a basic notion applicable to all models of computation. Direct sum theorems and XOR lemmas for distributional complexity, in addition to being statements of independent interest, have found applications in areas ranging from derandomization~\cite{Yao82,NW94,IW97} to streaming~\cite{AN21} and property testing~\cite{BKST23}.

Let $f: \bn \to \bits$  be a boolean function, $\mu$ be a distribution over $\bn$, and consider the task of computing the $k$-fold direct product $f^{\otimes k}(X^{(1)},\ldots,X^{(k)}) \coloneqq (f(X^{(1)}),\ldots,f(X^{(k)}))$ of $f$ to error $\varepsilon$ with respect to~$\mu^k$. One strategy is to sequentially compute each $f(X^{(i)})$ to error $\varepsilon/k$ with respect to $\mu$ and apply the union bound. Writing $\overline{\mathrm{Depth}}^{\mu}(f,\varepsilon)$ to denote the minimum expected depth of any decision tree that computes $f$ to error $\varepsilon$ w.r.t.~$\mu$, this shows that: 
\begin{equation*}
\overline{\mathrm{Depth}}^{\mu^k}(f^{\otimes k},\varepsilon) \le k \cdot \overline{\mathrm{Depth}}^{\mu}(f,\tfrac{\varepsilon}{k}).
\label{eq:naive upper bound}
\end{equation*}

Our main result is that this naive strategy is essentially optimal for all functions and distributions:\medskip 

\begin{tcolorbox}[colback = white,arc=1mm, boxrule=0.25mm]
\begin{theorem}[Strong direct sum theorem for distributional query complexity; special case of~\Cref{thm:main-formal}]
\label{thm:main-intro}
For every function $f: \bn \to \bits$, distribution $\mu$  over $\bn$, integer $k \in \N$, and $\varepsilon < 1$, 
\[ \overline{\mathrm{Depth}}^{\mu^k}(f^{\otimes k},\varepsilon) \ge \tilde{\Omega}(\eps^2k) \cdot \overline{\mathrm{Depth}}^{\mu}(f,\Theta(\tfrac{\varepsilon}{k})).\]
\end{theorem}
\end{tcolorbox}
\medskip 

Such direct sum theorems are termed {\sl strong}, referring to the fact that they show that computing a direct product necessitates a blowup in {\sl both} the computational resources of interest---in our case, query complexity---{\sl and} error. Strong direct sum theorems contrast with standard ones, which only show a blowup in computational resource, and also with direct {\sl product} theorems, which focus  on the blowup in error. We give a detailed overview of prior work in~\Cref{sec:background}, mentioning for now that while standard direct sum and direct product theorems for distributional query complexity have long been known, a strong direct sum theorem had been elusive.  Prior to our work, it was even open whether 
$\overline{\mathrm{Depth}}^{\mu^k}(f^{\otimes k},1.01\varepsilon) \ge 1.01 \cdot \overline{\mathrm{Depth}}^{\mu}(f,\varepsilon)$ holds. Indeed, the problem is known to be quite subtle, as a striking counterexample of Shaltiel~\cite{Sha04} shows that a strong direct sum
theorem is badly false if one considers {\sl worst-case} instead of expected query complexity.

\paragraph{A strong XOR lemma.} We also obtain a strong XOR lemma (\Cref{thm:main-xor}) as a corollary of a simple equivalence between direct sum theorems and XOR lemmas for query complexity. One direction is immediate, since the $k$-fold XOR $f^{\oplus k}(X^{(1)},\ldots,X^{(k)}) \coloneqq f(X^{(1)}) \oplus \cdots \oplus f(X^{(k)})$ can only be easier to compute than the $k$-fold direct product.
For the query model we show that the converse also holds: a direct sum theorem implies an XOR lemma with analogous
parameters. 



\section{Broader context: Comparison with the randomized setting}
\label{sec:challenges}

Direct sum theorems are also well-studied in the setting of {\sl randomized} query complexity. Recall that the $\eps$-error randomized query complexity of $f$, denoted $\overline{\mathrm{R}}(f,\eps)$, is the minimum expected depth of any randomized  decision tree that computes $f$ with error at most $\eps$ for all inputs.  By Yao's minimax principle, direct sum theorems for distributional query complexity imply analogous ones for randomized query complexity. However, as we now elaborate, such theorems are substantially more difficult to prove in the distributional setting.

\subsection{A simple and near-optimal strong direct sum theorem for $\overline{\mathrm{R}}$} 
\label{sec:simple}
For randomized query complexity, proving a strong direct sum theorem with near-optimal parameters requires only two observations. 

\paragraph{Observation \#1.} The first is that a standard direct sum theorem, one without error amplification, easily holds in the distributional setting, and hence the randomized setting as well by Yao's principle:
\begin{equation}
\overline{\mathrm{Depth}}^{\mu^k}(f^{\otimes k},\varepsilon) \ge \Omega(k) \cdot \overline{\mathrm{Depth}}^{\mu}(f,\varepsilon) \quad \text{and therefore} \quad \overline{\mathrm{R}}(f^{\otimes k},\eps) \ge \Omega(k)\cdot \overline{\mathrm{R}}(f,\eps). \label{eq:BDK1} \end{equation} 
The idea is that given a decision tree $T$ of average depth $q$ that computes $f^{\otimes k}$ with error~$\eps$, one can extract a decision tree of average depth $q/k$ that computes $f$ to error $\eps$:  place the input in a random block $\boldsymbol{i}\sim [k]$, fill the remaining blocks with independent random draws from $\mu$, and return the $\boldsymbol{i}$th bit of $T$'s output. It is straightforward to show that this reduces the average depth of $T$ by a factor of $k$ while preserving its error. 

\paragraph{Observation \#2.} The second observation is standard error reduction of randomized algorithms by repetition,  which in particular implies:  
\begin{equation} \overline{\mathrm{R}}(f, \lfrac{\eps}{k}) \leq O\paren*{\log k}\cdot \overline{\mathrm{R}}(f,\eps).\label{eq:success-amplification}
\end{equation} 
Combining~\Cref{eq:BDK1,eq:success-amplification} yields a strong direct sum theorem
\begin{equation*}
    \overline{\mathrm{R}}(f^{\otimes k}, \eps) \geq \Omega\paren*{\frac{k}{\log k}} \cdot \overline{\mathrm{R}}(f,\tfrac{\eps}{k})\label{eq:easy}
\end{equation*}
that is within a $O(\log k)$ factor of optimal.  

\paragraph{Blais--Brody.} Using more sophisticated techniques, Blais and Brody \cite{BB19} were recently able remove to this $O(\log k)$ factor and obtain an optimal strong direct sum theorem for $\overline{\mathrm{R}}$. Building on their work, Brody, Kim, Lerdputtipongporn, and Srinivasulu~\cite{BKLS23} then obtained an optimal strong XOR lemma for $\overline{\mathrm{R}}$.


\subsection{Error reduction fails in the distributional setting} 

While the crux of~\cite{BB19} and~\cite{BKLS23}'s works is the removal of a $O(\log k)$ factor, the situation is very different in the distributional setting. As mentioned, prior to our work even a direct sum theorem where both factor-of-$\tilde{\Omega}(k)$ blowups in~\Cref{thm:main-intro} are replaced by 1.01 was not known to hold.

With regards to the argument above, it is Observation \#2 that breaks in the distributional setting---not only does error reduction by repetition break, the distributional analogue of~\Cref{eq:success-amplification} is simply false. This points to a fundamental difference between distributional and randomized complexity: while generic error reduction of randomized algorithms is possible in all reasonable models of computation, the analogous statement for distributional complexity is badly false in all reasonable models of computation. For the query model specifically, in~\Cref{sec:no-boosting} we give an easy proof of the following: 

\begin{fact}
    \label{fact:no-boosting}
    For any $n \in \N$ and $\mu$ being the uniform distribution over $\bits^n$, there is a function $f:\bits^n \to \bits$ such that $\overline{\mathrm{Depth}}^{\mu}(f,\tfrac{1}{4}) = 0$ and yet $\overline{\mathrm{Depth}}^{\mu}(f,\tfrac{1}{8}) \geq \Omega(n).$
\end{fact}

\subsection{A brief summary of our approach}
We revisit Observation \#1 and show how the very same extraction strategy can in fact yield a tree with error $\Theta(\eps/k)$, instead of~$\eps$, at the expense of only a slight increase in depth.  A key technical ingredient in our analysis is {\sl Impagliazzo's Hardcore Theorem}~\cite{Imp95}. For intuition as to why this theorem may be relevant for us, we note that it is tightly connected to the notion of {\sl boosting} from learning theory---they are, in some sense, dual to each other~\cite{KS03}. And boosting is, of course, a form of error reduction, albeit one that is more intricate than error reduction by repetition. 

See~\Cref{sec:tech-overview} for a detailed overview of our approach, including a discussion of why Impagliazzo's Hardcore Theorem, as is, does not suffice, thereby necessitating our new ``resilience lemma" that accompanies it.

\section{Prior Work}
\label{sec:background}

We now place~\Cref{thm:main-intro} within the context of prior work on direct sum and product theorems for distributional query complexity. This is a fairly large body of work that dates back to the 1990s. 

\subsection{Standard direct sum and product theorems}

\paragraph{Standard direct sum theorems.} 

As we sketched in~\Cref{sec:simple}, a simple argument shows that 
\begin{equation}
\overline{\mathrm{Depth}}^{\mu^k}(f^{\otimes k},\varepsilon) \ge \Omega(k) \cdot \overline{\mathrm{Depth}}^{\mu}(f,\varepsilon). \label{eq:BDK}
\end{equation} 
This along with an application of Markov's inequality yields: 
\begin{equation}
\mathrm{Depth}^{\mu^k}(f^{\otimes k},\varepsilon-\varepsilon') \ge \Omega(\varepsilon'k)\cdot \mathrm{Depth}^{\mu}(f,\varepsilon),    \label{eq:JKS}
\end{equation}
where $\mathrm{Depth}^{\mu}(\cdot ,\cdot)$ is the analogue of $\overline{\mathrm{Depth}}^{\mu}(\cdot,\cdot)$ for worst-case instead of expected query complexity. (The details of these arguments are spelt out in~\cite{JKS10,BDK18}.)

Note that the error budget is the same on both sides of~\Cref{eq:BDK} and the error budget on the RHS of~\Cref{eq:JKS} is {\sl larger} than that of the LHS. In a strong direct sum theorem one seeks a lower bound even when the error budget on the RHS is much smaller than that of the LHS, ideally by a multiplicative factor of $k$ to match the naive upper bound.

\paragraph{A direct product theorem.} Impagliazzo, Raz, and Wigderson~\cite{IRW94} proved a direct product theorem which focuses on the blowup in error.  They showed that: 
\begin{equation} \mathrm{Depth}^{\mu^k}(f^{\otimes k},\varepsilon) \ge  \mathrm{Depth}^{\mu}(f,\tfrac{\varepsilon}{k}).\label{eq:IRW}
\end{equation}
While this result has the sought-for factor of $k$ difference between the error budgets on the LHS and RHS, it comes at the price of there no longer being {\sl any} blowup in depth. 
 




\subsection{Progress and barriers towards a strong direct sum theorem}

These results naturally point to the problem of proving a unifying strong direct sum theorem. We now survey efforts at such a best-of-both-worlds result over the years. 

\paragraph{Decision forests.} Nisan, Rudich, and Saks~\cite{NRS94} proved the following strengthening of~\cite{IRW94}’s result. While~\cite{IRW94} gives an upper bound on the success probability of a {\sl single} depth-$d$ decision tree for $f^{\otimes k}$,~\cite{NRS94} showed that the same bound holds even for {\sl decision forests} where one gets to construct a different depth-$d$ tree for each of the $k$ copies of $f$.     

Since one can always stack the $k$ many depth-$d$ trees in a decision forest to obtain a single tree of depth~$kd$,~\cite{NRS94}'s result establishes a special case of a strong direct sum theorem under a structural assumption on the tree for $f^{\otimes k}$. See~\Cref{fig:stacked} in~\Cref{app:figures} for an illustration of the stacked decision tree that one gets from a decision forest.

\paragraph{Fair decision trees.} Building on the techniques of~\cite{NRS94}, Shaltiel~\cite{Sha04} proved a strong direct sum theorem under a different structural assumption on the tree for $f^{\otimes k}$. He considered decision trees of depth $kd$ that are “fair” in the sense that every path queries each of the $k$ blocks of variables at most $d$ times. (\cite{Sha04} actually proved a strong XOR lemma for fair decision trees, which implies a strong direct sum theorem for such trees.) See~\Cref{fig:fair} in~\Cref{app:figures} for an illustration of a fair decision tree. 


\paragraph{Shaltiel's counterexample for worst-case query complexity.}

These results of~\cite{NRS94} and~\cite{Sha04} could be viewed as evidence in favor of a general strong direct sum theorem, one that does not impose any structural assumptions on the tree for $f^{\otimes k}$. However, in the same paper Shaltiel also presented an illuminating example: he constructed a function, which we call $\mathsf{Shal}$, and a distribution $\mu$ such that for all $k\in \N$,
\begin{equation} 
\mathrm{Depth}^{\mu^k}(\mathsf{Shal}^{\otimes k},\varepsilon) \le O\big(\mathrm{Depth}^{\mu}(\mathsf{Shal},\tfrac{\varepsilon}{k})\big).
\label{eq:shaltiel}
\end{equation}
This shows, surprisingly, that for {\sl worst-case} query complexity, the factor-of-$\Omega(k)$ blowup in query complexity that one seeks in a strong direct sum theorem is not always necessary, and in fact sometimes even a constant factor suffices. 

\paragraph{Shaltiel's counterexample vs.~\Cref{thm:main-intro}.} This counterexample for worst-case query complexity should be contrasted with our main result,~\Cref{thm:main-intro}, which shows that a strong direct sum theorem holds for {\sl expected} query complexity. 
Indeed, the starting point of our work was the encouraging observation that Shaltiel's function does in fact satisfy a strong direct sum theorem if one instead considers expected query complexity. That is, for any $\eps < 1$ and sufficiently large $k$,
\[ \overline{\mathrm{Depth}}^{\mu^k}(\mathsf{Shal}^{\otimes k},\varepsilon) \ge \Omega(k) \cdot\overline{\mathrm{Depth}}^{\mu}(\mathsf{Shal},\tfrac{\eps}{k}).\]
This is a simple observation but appears to have been overlooked. As we now overview, subsequent work considered other ways of sidestepping Shaltiel's counterexample. 



\subsection{Results in light of Shaltiel's counterexample}

\paragraph{A strong direct sum theorem for the OR function.} Klauck, {\v{S}}palek, and de Wolf~\cite{KSdW07} sidestepped Shaltiel’s counterexample by considering a {\sl specific} function (and distribution): motivated by applications to time-space tradeoffs, they proved a strong direct sum theorem for the OR function and with $\mu$ being its canonical hard distribution. Using this, they also showed, for all functions $f$ a lower bound on $f^{\otimes k}$'s query complexity in terms of $f$’s block sensitivity. This stands in contrast to a strong direct sum theorem where one seeks a lower bound on $f^{\otimes k}$'s query complexity in terms of~$f$’s query complexity.

\paragraph{A phase transition in Shaltiel's counterexample.} 
The precise parameters of Shaltiel's counterexample are:
\[ 
\mathrm{Depth}^{\mu^k}(\mathsf{Shal}^{\otimes k},e^{-\Theta(\delta k)}) \le C\delta k \cdot \mathrm{Depth}^{\mu}(\mathsf{Shal},\delta)
\] 
for all sufficiently large constants $C$. Importantly, the multiplicative factor on the RHS is only $\delta k$ instead of $k$, and therefore becomes a constant if the initial hardness parameter is $\delta = \varepsilon/k$ (thereby yielding~\Cref{eq:shaltiel}).

Drucker~\cite{Dru12} showed that there is a ``phase transition" in Shaltiel’s counterexample in the following sense: for all functions $f$ and a sufficiently small constant $c > 0$, 
\begin{equation}
\mathrm{Depth}^{\mu^k}(f^{\otimes k},1-e^{-\Theta(\delta k)}) \ge c\delta k \cdot \mathrm{Depth}^{\mu}(f,\delta).
\label{eq:drucker}
\end{equation}
Therefore, while~\cite{Sha04} showed the existence of a function $\mathsf{Shal}$ such that its $k$-fold direct product can be computed to surprisingly low error if the depth budget is $C\delta k\cdot \mathrm{Depth}^{\mu}(f,\delta)$ for a sufficiently {\sl large} constant $C$,~\cite{Dru12} showed that for all functions $f$, this stops being the case if the depth budget is instead $c\delta k\cdot \mathrm{Depth}^{\mu}(f,\delta)$ for a sufficiently {\sl small} constant $c$.

\paragraph{Query complexity with aborts.} Blais and Brody~\cite{BB19} showed that Shaltiel's counterexample can be sidestepped in a different way. En route to proving their strong direct sum theorem for randomized query complexity (discussed in~\Cref{sec:challenges}), they considered decision trees $T: \bn \to \{\pm 1,\bot\}$ that are allowed to output $\bot$ (``abort") on certain inputs, and where the error of $T$ in computing a function $f: \bn \to \bits$ is measured with respect to $T^{-1}(\{\pm 1\})$. In other words, $T$’s output on $x$ is considered correct if $T(x) = \bot$. 

Writing $\mathrm{Depth}^{\mu}_{\Pr[\bot]\le\frac1{3}}(f,\varepsilon)$ to denote the minimum depth of any decision tree for $f$ that aborts with probability at most $1/3$ and otherwise errs with probability at most $\varepsilon$ (both w.r.t.~$\mu$),~\cite{BB19} proved that 
\begin{equation}
\mathrm{Depth}^{\mu^k}_{\Pr[\bot]\le\frac1{3}}(f^{\otimes k},\varepsilon) \ge \Omega(k) \cdot \mathrm{Depth}^{\mu}_{\Pr[\bot]\le\frac1{3}}(f,\tfrac{\varepsilon}{k}).
\label{eq:blais brody}
\end{equation}

Even though the error budget on non-aborts is only $\eps/k$ on the RHS, the fact that the tree is allowed to abort with probability $1/3$ means that it is deemed correct on a $1/3$ fraction of inputs ``for free". A decision tree that aborts with probability $1/3$ and otherwise errs with probability $\eps/k$ can therefore be much smaller than one that never aborts and errs with probability $\eps/k$, and indeed, it is easy to construct examples witnessing the maximally large separation: 
\[ 
n = \mathrm{Depth}^{\mu}(f,\tfrac{\varepsilon}{k})  \gg \mathrm{Depth}^{\mu}_{\Pr[\bot]\le\frac1{3}}(f,\tfrac{\varepsilon}{k}) = 1.
\]

Building on~\cite{BB19}, Brody, Kim, Lerdputtipongporn, and Srinivasulu~\cite{BKLS23} proved a strong XOR lemma for this model of query complexity with aborts, achieving analogous parameters.  

\paragraph{A strong XOR lemma assuming hardness against all depths.}
A standard strong XOR lemma states that if $f$ is hard against decision trees of certain {\sl fixed} depth $d$, then $f^{\otimes k}$ is much harder against decision trees of depth $\Omega(dk)$. Recent work of Hoza~\cite{Hoza23} shows that Shaltiel's counterexample can be sidestepped if one allows for the stronger assumption that $f$'s hardness ``scales nicely" with $d$. (See the paper for the precise statement of the resulting strong XOR lemma.)

\subsection{Summary}

Summarizing, prior work on direct sum and product theorems for distributional query complexity either: focused on the blowup in error~\cite{IRW94} or query complexity~\cite{JKS10, BDK18}  but not both; considered restrictions (fair decision trees~\cite{Sha04}) or variants (decision forests~\cite{NRS94}; allowing for aborts~\cite{BB19, BKLS23}) of the query model; focused on specific functions (the OR function~\cite{KSdW07}); or imposed additional hardness assumptions about the function~\cite{Hoza23}. \Cref{thm:main-intro}, on the other hand, gives a strong direct sum theorem that holds for all functions in the standard query model. See~\Cref{table}. 
\medskip

\begin{table}[ht]
\captionsetup{width=.9\linewidth}
\renewcommand{\arraystretch}{1.9}
\centering
\begin{tabular}{|c|c|c|c|}
\hline
 ~~~Reference~~~ &  \begin{tabular}{c} Error \vspace{-12pt} \\ ~~Amplification~~ \end{tabular} &  \begin{tabular}{c} Query \vspace{-12pt} \\ ~~Amplification~~ \end{tabular}  & \begin{tabular}{c} Query model/ \vspace{-10pt} \\
 Assumption \end{tabular} 
 \\
\hline 
\hline 
~~\cite{JKS10,BDK18}~~ & $\times$ & $\checkmark$ & Standard query model  \\ \hline 
\cite{IRW94} & $\checkmark$ & $\times$ & Standard query model \\ \hline 
\hline 
\cite{NRS94} & $\checkmark$ & $\checkmark$ & Decision forests \\ \hline 
\cite{Sha04} & $\checkmark$ & $\checkmark$ &~~~Fair decision trees~~~\\ \hline 
\cite{KSdW07} & $\checkmark$ & $\checkmark$ & $f = \mathrm{OR}$ \\ \hline 
\cite{BB19,BKLS23} & $\checkmark$ & $\checkmark$ & ~~Decision trees with aborts~~ \\ \hline 
\cite{Hoza23} & $\checkmark$ & $\checkmark$ & ~~Hardness against all depths~~ \\ \hline \hline 
\Cref{thm:main-intro} & $\checkmark$ & $\checkmark$ & Standard query model \\ \hline 
\end{tabular} 
\caption{Direct sum and product theorems for distributional query complexity}
\label{table} 
\end{table}

\section{Formal statements of our results and their tightness}

\Cref{thm:main-intro} is a special case of the following result: 

\begin{theorem}[Strong direct sum theorem]
\label{thm:main-formal} 
For every function $f: \bn \to \bits$, distribution $\mu$  over $\bn$, $k\in \mathbb{N}$, and $\gamma, \delta \in (0,1)$, we have that 
\[ 
 \overline{\mathrm{Depth}}^{\mu^k}(f^{\otimes k},1-e^{-\Theta(\delta k)}-\gamma) \ge \Omega\left( \frac{\gamma^2 k}{\log(1/\delta)}\right) \cdot \overline{\mathrm{Depth}}^{\mu}(f,\delta).
\] 
\end{theorem} 
We in fact prove a strong {\sl threshold} direct sum theorem which further generalizes~\Cref{thm:main-formal}: while a direct sum theorem shows that $f^{\otimes k}$ is hard to compute, i.e.~it is hard to get {\sl all} $k$ copies of $f$ correct, a threshold direct sum theorem shows that it is hard even to get {\sl most} of the $k$ copies of $f$ correct. See~\Cref{thm:threshold-DPT}. 

By the equivalence between strong direct sum theorems and strong XOR lemmas (\Cref{claim:equivalence between direct product and XOR}), we also get: 

\begin{theorem}[Strong XOR lemma]
\label{thm:main-xor}
For every function $f:\bits^n\to\bits$ and distribution $\mu$ over $\bits^n$, $k\in\N$, and $\gamma,\delta\in(0,1)$, we have that
    $$
    \overline{\mathrm{Depth}}^{\mu^k}\left(f^{\oplus k},\tfrac{1}{2}(1-e^{-\Theta(\delta k)}-\gamma)\right)\ge {\Omega}\left(\frac{\gamma^2 k}{\log(1/\delta)}\right)\cdot\overline{\mathrm{Depth}}^\mu(f,\delta).
    $$
\end{theorem}

\subsection{Tightness} \Cref{thm:main-formal} amplifies an initial hardness parameter of $\delta = \Theta(1/k)$ to $1-\gamma$ for any small constant~$\gamma$ with a near-optimal overhead of 
\[ \Omega\left(\frac{\gamma^2 k}{\log(1/\delta)}\right) = \Omega\left(\frac{k}{\log k}\right).\]
However, due to the polynomial dependence on $\gamma$, we cannot achieve a final hardness parameter that is exponentially close to $1$ as a function of $k$. We show that this is unavoidable since at least a linear dependence on $\gamma$ is necessary:  

\begin{claim}[Linear dependence on $\gamma$ is necessary]
\label{claim:linear dependence on gamma}
        Let $\mathsf{Par} : \bn\to\bits$ be the parity function and $\mu $ be the uniform distribution over $\bn$. Then for all $\gamma$,
    \[
\overline{\mathrm{Depth}}^{\mu^k}(\mathsf{Par}^{\otimes k}, 1 -\gamma) \leq O(\gamma k) \cdot \overline{\mathrm{Depth}}^\mu(\mathsf{Par}, \tfrac1{4}).
    \]
\end{claim}

The same example shows that a linear dependence on $\gamma$ is likewise necessary in the setting of XOR lemmas.  Determining the optimal polynomial dependence on $\gamma$ in both settings, as well as the necessity of the $\log(1/\delta)$ factor, are concrete avenues for future work.

\section{Technical Overview for~\Cref{thm:main-formal}} 
\label{sec:tech-overview}


\subsection{Hardcore measures and the Hardcore Theorem} 

At the heart of our proof is the notion of a {\sl hardcore measure} and {\sl Impagliazzo’s Hardcore Theorem}~\cite{Imp95}, both adapted to the setting of query complexity.  

\begin{definition}[Hardcore measure for query complexity] 
\label{def:hardcore measure} 
We say that $H : \bn \to [0,1]$ is a {\sl $(\gamma,d)$-hardcore measure for $f : \bn \to \bits$ w.r.t.~$\mu$ of density $\delta$} if:
\begin{enumerate} 
\item $H$'s density is $\delta$: $\ds\Ex_{\bx\sim \mu}[H(\bx)] = \delta$. 
\item $d$-query algorithms achieve correlation at most $\gamma$ with $f$ on $H$:
\[ \Ex_{\bx\sim \mu}[f(\bx)T(\bx)H(\bx)] \le \gamma \Ex_{\bx \sim \mu}[H(\bx)] = \gamma \delta.\]
for all decision trees $T$ whose expected depth w.r.t.~$\mu$ is at most $d$.
\end{enumerate}
\end{definition}

\begin{theorem}[Hardcore Theorem for query complexity]
    \label{thm:hardcore theorem for DTs}
For every function $f : \bn\to\bits$, distribution $\mu$ over $\bn$, and $\gamma,\delta > 0$, there exists a $(\gamma,d)$-hardcore measure $H$ for $f$ of density $\delta/2$ w.r.t.~$\mu$ where 
\[ d = \Theta\left(\frac{\gamma^2}{\log(1/\delta)}\right)  \overline{\mathrm{Depth}}^\mu(f,\delta). \]
\end{theorem}


The Hardcore Theorem was originally proved, and remains most commonly used, in the setting of circuit complexity where it has long been recognized as a powerful result. (See e.g.~\cite{Tre07}, where it is described as ``one of the bits of magic of complexity theory".) We show in~\Cref{appendix:hardcore} that its proof extends readily to the setting of query complexity to establish~\Cref{thm:hardcore theorem for DTs}.  Despite its importance in circuit complexity and its straightforward extension to query complexity, our work appears to be the first to consider its applicability in the latter setting.

\begin{remark}
For intuition regarding~\Cref{def:hardcore measure}, note that if $H : \bn \to \zo$ is the indicator of a {\sl set}, the two properties simplify to: 
 $\ds \Prx_{\bx \sim \mu}[\bx \in H] = \delta$ and 
$ \ds \Ex_{\bx\sim \mu}[f(\bx)T(\bx) \mid \bx \in H] \le \gamma$. 
    \end{remark}

\subsection{Two key quantities: hardcore density and hardcore advantage at a leaf}

\paragraph{Setup.} For the remainder of this section, we fix a function $f : \bn\to\bits$, distribution $\mu$ over $\bn$, and initial hardness parameter $\delta$ (which we think of as small, close to $0$). Let $T : (\bn)^k \to \bits^k$ be a decision tree that seeks to compute $f^{\otimes k}$ w.r.t.~$\mu^k$.   Our goal is to show that $T$'s error must be large, close to $1$, unless its depth is sufficient large.

\paragraph{Definitions of the hardcore density and hardcore advantage at a leaf.} Let $H : \bn\to [0,1]$ be a $(\gamma,d)$-hardcore measure for $f$ of density $\delta$ w.r.t.~$\mu$ given by~\Cref{thm:hardcore theorem for DTs}.  Each leaf $\ell$ of $T$ corresponds to a tuple of restrictions $(\pi_1,\ldots,\pi_k)$ to each of the $k$ blocks of inputs. We will be interested in understanding, for a random block $\boldsymbol{i}\in [k]$, the extent to which the restricted function $H_{\pi_{\boldsymbol{i}}}$ retains the two defining properties of a hardcore measure: high density and strong hardness. We therefore define:

\begin{definition}[Hardcore density at $\ell$]
\label{def:hardcore-density}
For $i\in [k]$, the {\sl hardcore density at $\ell$ in the $i$th block} is the quantity: 
\[ \Dens_H(\ell,i)\coloneqq \Ex_{\bX\sim \mu^k}\big[H(\bX^{(i)})\mid \text{$\bX$ reaches $\ell$}\big]. \]
The {\sl total hardcore density at $\ell$} is the quantity $\ds \Dens_H(\ell) \coloneqq \sum_{i =1}^k \Dens_H(\ell,i).$
\end{definition}

See \Cref{fig:hardcore} for an illustration of \Cref{def:hardcore-density}.

\def\centerarc[#1](#2)(#3:#4:#5){ \draw[#1] ($(#2)+({#5*cos(#3)},{#5*sin(#3)})$) arc (#3:#4:#5); }
\tikzset{
    ncbar angle/.initial=90,
    ncbar/.style={
        to path=(\tikztostart)
        -- ($(\tikztostart)!#1!\pgfkeysvalueof{/tikz/ncbar angle}:(\tikztotarget)$)
        -- ($(\tikztotarget)!($(\tikztostart)!#1!\pgfkeysvalueof{/tikz/ncbar angle}:(\tikztotarget)$)!\pgfkeysvalueof{/tikz/ncbar angle}:(\tikztostart)$)
        -- (\tikztotarget)
    },
    ncbar/.default=0.5cm,
}

\tikzset{round left paren/.style={ncbar=0.5cm,out=120,in=-120}}
\tikzset{round right paren/.style={ncbar=0.5cm,out=60,in=-60}}

\newcommand{\quartercirc}{%
\begin{tikzpicture}
    \centerarc[black,fill=gray!30](0,0)(0:90:0.24);
    \filldraw[fill=gray!30,opacity=0.0,fill opacity=1] (0,0) -- (0,0.24) -- (0.24,0) -- cycle;
    \draw[-] (0,0) to (0,0.24);
    \draw[-] (0,0) to (0.24,0);
\end{tikzpicture}%
}

\begin{figure}[h!]
    \centering
    \begin{tikzpicture}[tips=proper]
        \def\inc{1.5}
        \def\sinc{1.2}
        \def\s{0.75}
    
        \node[isosceles triangle,
            draw,
            thick,
            isosceles triangle apex angle=60,
            rotate=90,
            minimum size=6cm] (T1) at (0,0){};
        
        \draw[black,dashed] (T1.east) .. controls ([xshift=-0.1cm]T1.358) .. ([yshift=-0.5cm]T1.east) node[] (N1) {{}};
        \draw[black,dashed] ([yshift=-0.5cm]T1.east) .. controls ([xshift=0.4cm]T1.40) and (T1.350) .. (T1.center) node[fill=white,pos=0.9] (N1) {$\pi$};
        \draw[black,dashed] (T1.center) .. controls ([xshift=2.4cm]T1.110) .. (T1.west) node[] (N2) {{}};

        \coordinate (top) at (5,4);
        \coordinate (top2) at ($(top)+(\inc,0)$); 
        \coordinate (top3) at ($(top2)+(\inc,0)$);
        
        \node[rectangle,
    	draw = black,
    	minimum width = 1cm, 
    	minimum height = 1cm] (r) at (top) {};
        \node[ellipse,
            draw = black,
            fill=gray!30,
            minimum width = 0.4cm,
            minimum height = 0.4cm] (e) at (top) {\footnotesize $H$};
        \node[rectangle,
    	draw = black,
    	minimum width = 1cm, 
    	minimum height = 1cm] (r2) at (top2) {};
        \node[ellipse,
            draw = black,
            fill=gray!30,
            minimum width = 0.5cm,
            minimum height = 0.5cm] (e2) at (top2) {\footnotesize $H$};
        \node[rectangle,
    	draw = black,
    	minimum width = 1cm, 
    	minimum height = 1cm] (r3) at (top3) {};
        \node[ellipse,
            draw = black,
            fill=gray!30,
            minimum width = 0.5cm,
            minimum height = 0.5cm] (e3) at (top3) {\footnotesize $H$};

        \draw [black] (4.3,3.4) to [round left paren ] (4.3,4.6);
        \draw [black] (8.7,3.4) to [round right paren ] (8.7,4.6);

        \draw [black] (5.75,3.5) node [] {\LARGE ,};
        \draw [black] (7.25,3.5) node [] {\LARGE ,};

        \draw [black] (4.8,-2.5) to [round left paren ] (4.8,-1.5);
        \draw [black] (8.4,-2.5) to [round right paren ] (8.4,-1.5);

        \draw [black] (6,-2.4) node [] {\LARGE ,};
        \draw [black] (7.2,-2.4) node [] {\LARGE ,};

        \coordinate (bot) at (5.4,-2); 
        \coordinate (bot2) at ($(bot)+(\sinc,0)$);
        \coordinate (bot3) at ($(bot2)+(\sinc,0)$);
        
        \node[rectangle,
    	draw = black,
    	minimum width = \s cm, 
    	minimum height = \s cm] (r2) at (bot) {};
        \node[rectangle,
    	draw = black,
    	minimum width = \s cm, 
    	minimum height = \s cm] (r3) at (bot2) {};
        \node[rectangle,
    	draw = black,
    	minimum width = \s cm, 
    	minimum height = \s cm] (r4) at (bot3) {};
       
        \centerarc[black,fill=gray!30]([xshift=-0.38cm,yshift=-0.38cm]bot)(0:90:0.5);
        \centerarc[black,fill=gray!30]([xshift=0.38cm,yshift=0.38cm]bot2)(180:270:0.5);
        \centerarc[black,fill=gray!30]([xshift=0.38cm,yshift=-0.38cm]bot3)(90:180:0.5);

        \filldraw[fill=gray!30,opacity=0.0,fill opacity=1] (r2.225) -- (r2.162) -- (r2.288) -- cycle;
        \filldraw[fill=gray!30,opacity=0.0,fill opacity=1] (r3.342) -- (r3.45) -- (r3.108) -- cycle;
        \filldraw[fill=gray!30,opacity=0.0,fill opacity=1] (r4.251) -- (r4.315) -- (r4.18) -- cycle;

        \draw[color=black] (T1.west) node [below,fill=white] {\footnotesize Leaf $\ell$};
        \node[draw,circle,fill=black,inner sep=1pt] (x) at (T1.west) {};


        \draw[-{Stealth[scale=1]}] (4.6,-3.3) to ([xshift=-.2 cm,yshift=-.2cm]bot);
        \draw [black] (4.6,-3.3) node [fill=white] {Hardcore measure $H_{\pi_1}$ with $\mathrm{Dens}_H(\ell, 1)={\mathrm{area}(\quartercirc)}/{\mathrm{area}(\square})$};

        \draw[-{Stealth[scale=1]}] (6.5,3) to node[midway,fill=white!30,scale=1] {Queries along path $\pi$} (6.5,-1);

        \node[rectangle,
    	draw = black,
    	minimum width = \s cm, 
    	minimum height = \s cm] (aaaa) at (bot) {};
        \node[rectangle,
    	draw = black,
    	minimum width = \s cm, 
    	minimum height = \s cm] (bbbbb) at (bot2) {};
        \node[rectangle,
    	draw = black,
    	minimum width = \s cm, 
    	minimum height = \s cm] (ccccc) at (bot3) {};
        
    \end{tikzpicture}
    \captionsetup{width=.9\linewidth}
    \caption{Illustration of a hardcore density. The tree $T:(\bn)^3\to\bits^3$ seeks to compute a function $f^{\otimes 3}$. The tuple of squares at the top of the figure illustrates the set of all inputs to the function while the strings in the support of the hardcore measure are shaded gray. The tuple at the bottom of the figure illustrates the set of inputs reaching the leaf $\ell$. Each block is the subcube consistent with the path $\pi$ and the shaded region denotes the fragment of $H$ which is contained in the corresponding subcube. }
    \label{fig:hardcore}
\end{figure}
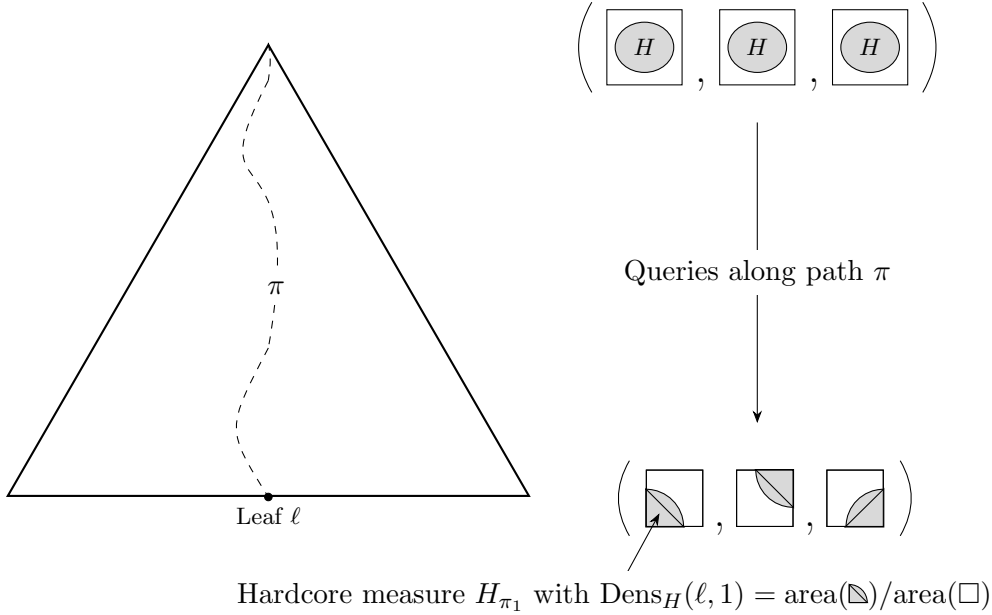

\begin{definition}[Hardcore advantage at $\ell$]
\label{def:hardcore advantage}
For $i\in [k]$, the {\sl hardcore advantage at $\ell$ in the $i$th block} is the quantity:
\[ \Adv_H(\ell,i) \coloneqq \abs[\Big]{\Ex_{\bX\sim \mu^k}\big[ f(\bX^{(i)})T(\bX)_i H(\bX^{(i)})\mid \text{$\bX$ reaches $\ell$}\big]}.  \]
The {\sl total hardcore advantage at $\ell$} is the quantity  
$ \ds \Adv_H(\ell) \coloneqq \sum_{i=1}^k \Adv_H(\ell, i). $
\end{definition}




Intuitively, leaves for which $\Dens_H(\ell)$ is large and $\Adv_H(\ell)$ is small contribute significantly to error of $T$.~\Cref{lem:accuracy in terms of hardcore density and advantage} below  formalizes this:

\paragraph{Notation} \hspace{-10pt} (Canonical distribution over leaves){\bf .}   We write $\mu^k(T)$ to denote the distribution over leaves of $T$ where:  
\[ \Prx_{\bell\sim \mu^k(T)}[\bell = \ell] = \Prx_{\bX\sim \mu^k}[\text{$\bX$ reaches $\ell$}].\]

\begin{tcolorbox}[colback = white,arc=1mm, boxrule=0.25mm]
\begin{lemma}[Accuracy in terms of hardcore density and advantage at leaves]
\label{lem:accuracy in terms of hardcore density and advantage}
\[ 
         \Prx_{\bX \sim \mu^k}[T(\bX) = f^{\otimes k}(\bX)] \leq \Ex_{\bell \sim \mu^k(T)}\bracket*{\exp\paren*{-\frac{\Dens_H(\bell) - \Adv_H(\bell)}{4}}}.
  \] 
\end{lemma} 
\end{tcolorbox}

\subsection{Expected total hardcore density and advantage}

\Cref{lem:accuracy in terms of hardcore density and advantage} motivates understanding the random variables $\mathrm{Dens}_H(\boldsymbol{\ell})$ and $\mathrm{Adv}_H(\boldsymbol{\ell})$ for $\boldsymbol{\ell}\sim \mu^k(T)$. We begin by bounding their expectations:

\begin{claim}[Expected total hardcore density]
\label{claim:expected total hardcore density} If $H$ is a hardcore measure of density $\delta$ then 
\[  \Ex_{\bell\sim \mu^k(T)}[\Dens_H(\bell)] = \delta k. \] 
\end{claim}

\Cref{claim:expected total hardcore density} is a statement about density preservation. It says that $H$'s expected density at a random leaf $\bell \sim \mu^k(T)$ and in a random block $\bi\sim [k]$ is equal to $H$'s initial density: 
\[  \mathop{\Ex_{\bell\sim \mu^k(T)}}_{\bi\sim [k]}\bigg[  \underbrace{\Ex_{\bX\sim \mu^k}\big[H(\bX^{(\bi)})\mid \text{$\bX$ reaches $\bell$}\big]}_{\Dens_H(\bell,\bi)}\bigg] = \delta = \Ex_{\bx\sim \mu}[H(\bX)]. \]

\begin{claim}[Expected total hardcore advantage]
\label{claim:expected total hardcore advantage}
    If $H$ is a $(\gamma, d)$-hardcore measure for $f$ of density $\delta$ w.r.t.~$\mu$ and the expected depth of $T$ is at most $dk$, then 
    \begin{equation*}
        \Ex_{\bell \sim \mu^k(T)}[\adv_H(\bell)] \leq \gamma \Ex_{\bell\sim\mu^k(T)}[\Dens_H(\bell)]. 
    \end{equation*}
\end{claim}

\Cref{claim:expected total hardcore advantage} is a statement about depth amplification. By definition, $H$ being a $(\gamma,d)$-hardcore measure for $f$ means that 
    \begin{equation*}
        \underbrace{ \Ex_{\bx \sim \mu}[ f(\bx) T_{\textnormal{small}}(\bx) H(\bx)]}_{\textnormal{Hardcore advantage}}  \leq \gamma  \underbrace{\Ex_{\bx \sim \mu}[H(\bx)]}_{\textnormal{Hardcore density}}
    \end{equation*}
    for every tree $T_{\mathrm{small}} : \bn\to\bits$ of expected depth $d$.~\Cref{claim:expected total hardcore advantage} says that 
\[ \underbrace{\sum_{i=1}^k \Ex_{\bX\sim \mu^k}\big[ f(\bX^{(i)})T_{\textnormal{large}}(\bX)_i H(\bX^{(i)})\big] }_{\textnormal{Total hardcore advantage}} \le  \gamma   \underbrace{\sum_{i=1}^k \Ex_{\bX\sim\mu^k}[H(\bX^{(i)})]}_{ \textnormal{Total hardcore density}}. \]
for every tree $T_{\textnormal{large}} : (\bn)^k \to \bits^k$ of expected depth $dk$. Crucially, the depth of $T_{\textnormal{large}}$ is allowed to be a factor of $k$ larger than that of $T_{\textnormal{small}}$, and yet the ratio of hardcore advantage to hardcore density remains the same ($\gamma$ in both cases).

\subsubsection{Done if Jensen went the other way}
\label{sec:jensen}
For intuition as to why~\Cref{claim:expected total hardcore density} and~\Cref{claim:expected total hardcore advantage} are relevant yet insufficient for us, note that if it were the case that $\mathbb{E}[\exp(-\boldsymbol{Z})] \le \exp(-\mathbb{E}[\boldsymbol{Z}])$, which unfortunately is the opposite of what Jensen’s inequality gives, we would have the strong bound on the accuracy of $T$ that we seek:   
\begin{align*}
 \Prx_{\bX \sim \mu^k}[T(\bX) = f^{\otimes k}(\bX)] &\leq \Ex_{\bell \sim \mu^k(T)}\bracket*{\exp\paren*{-\frac{\dens_H(\bell) - \adv_h(\bell)}{4}}} \tag{\Cref{lem:accuracy in terms of hardcore density and advantage}} \\
 &``\le" \exp\bigg(- \Ex_{\bell\sim\mu^k(T)}\bigg[\frac{\Dens_H(\bell)-\Adv_H(\bell)}{4}\bigg] \bigg) \tag{Wrong direction of Jensen} \\
 &\le \exp\bigg(-\frac{\delta k -\gamma \delta k}{4} \bigg)  \tag{\Cref{claim:expected total hardcore density} and \Cref{claim:expected total hardcore advantage}} \\
 &\le \exp(-\Theta(\delta k)). 
\end{align*}

For an actual proof, we need to develop a more refined understanding of the distribution of $\mathrm{Dens}_H(\boldsymbol{\ell})$ beyond just its expectation. (As it turns out, this along with the bound on $\mathbb{E}[\mathrm{Adv}_H(\boldsymbol{\ell})]$ given by~\Cref{claim:expected total hardcore advantage} suffices.) 

\subsection{A resilience lemma for hardcore measures}

\paragraph{An illustrative bad case to rule out.} Suppose $T$ were such that it achieved $\mathbb{E}[\mathrm{Dens}_H(\boldsymbol{\ell})] =  \delta k$ by having a $\delta$-fraction of leaves with $\mathrm{Dens}_H(\ell) = k$ and the remaining $1-\delta$ fraction with $\mathrm{Dens}_H(\ell) = 0$. If this were the case then “all the hardness” would be concentrated on a small $\delta$ fraction of leaves, and the best lower bound that we would be able to guarantee on error of $T$ with respect to $f^{\otimes k}$ would only be $\delta$. This is our starting assumption on the hardness of $f$, and so no error amplification has occurred. 

\paragraph{The resilience lemma.} We rule out cases like this by showing that $T$ must achieve $\mathbb{E}[\mathrm{Dens}_H(\boldsymbol{\ell})]=  \delta k$ by having the vast majority of its leaves with $\mathrm{Dens}_H(\ell) = \Omega(\delta k)$, i.e.~that $\Dens_H(\bell)$ is tightly concentrated around its expectation:  
\medskip 

\begin{tcolorbox}[colback = white,arc=1mm, boxrule=0.25mm]
\begin{lemma}[Resilience lemma]
\label{lem:resilience lemma}
For any hardcore measure $H$ of density $\delta$ w.r.t.~$\mu$ and tree $T:(\bn)^k \to \bits^k$,
\begin{equation*}
    \Prx_{\bell \sim \mu^k(T)}[\Dens_H(\bell) \leq \delta k/2] \leq e^{-\delta k/8}.
\end{equation*}
Similarly, $\ds \Prx_{\bell \sim \mu^k(T)}[\Dens_H(\bell) \geq 2\delta k] \leq e^{-\delta k/3}.$ 
\end{lemma}
\end{tcolorbox}
\medskip 

\begin{figure}[H]
\centering
\hspace*{\fill}
\begin{subfigure}{.45\textwidth}
  \centering
      \begin{tikzpicture}[tips=proper]
        \node[isosceles triangle,
            draw,
            isosceles triangle apex angle=60,
            rotate=90,
            minimum size=6cm] (T1) at (0,0){};
            
        \draw[black,dashed] (T1.east) .. controls ([xshift=-0.1cm]T1.358) .. ([yshift=-0.5cm]T1.east) node[] (N1) {{}};
        \draw[black,dashed] ([yshift=-0.5cm]T1.east) .. controls ([xshift=0.4cm]T1.40) and (T1.350) .. (T1.center) node[] (N1) {};
        \draw[black,dashed] (T1.center) .. controls ([xshift=2.4cm]T1.110) .. ([xshift=-0.15cm]T1.west) node[] (N2) {{}};
        
        \draw[-{Stealth[scale=1]}] ([xshift=0.5cm,yshift=0.5cm]T1.west) to ([xshift=-0.15cm]T1.west);
        \draw[] ([xshift=0.5cm,yshift=0.7cm]T1.west) node [right,black,fill=white, text width=2cm] {\footnotesize Leaf $\ell$ with $\mathrm{Dens}_H(\ell)=0$};
        
        \draw [black,decorate,decoration={brace,mirror,raise=1pt,amplitude=3pt}] ([xshift=0.2cm,yshift=-0.5cm]T1.left corner) -- ([xshift=2.8cm,yshift=-0.5cm]T1.left corner) node [black,pos=0.5,yshift=-0.5cm] {\small $\delta$-fraction};

        \draw [black,decorate,decoration={brace,mirror,raise=1pt,amplitude=3pt}] ([xshift=3.2cm,yshift=-0.5cm]T1.left corner) -- ([xshift=7.0cm,yshift=-0.5cm]T1.left corner) node [black,pos=0.5,yshift=-0.5cm] {\small $(1-\delta)$-fraction};

        \draw [] ([xshift=0cm]T1.left corner) -- (T1.right corner) node [black,pos=0.5,below] { $~ k\quad k\quad k\quad k\quad k \quad 0\quad 0\quad 0\quad 0\quad 0\quad 0\quad 0$};   
    \end{tikzpicture}
  \caption{An illustration of the bad case where $\mathrm{Dens}_H$ is anti-concentrated away from its mean of $\delta k$.}
  \label{fig:bad case}
\end{subfigure}%
\hfill
\begin{subfigure}{.45\textwidth}
  \centering
          \begin{tikzpicture}[tips=proper]
        \node[isosceles triangle,
            draw,
            isosceles triangle apex angle=60,
            rotate=90,
            minimum size=6cm] (T1) at (0,0){};
            
        \draw[black,dashed] (T1.east) .. controls ([xshift=-0.1cm]T1.358) .. ([yshift=-0.5cm]T1.east) node[] (N1) {{}};
        \draw[black,dashed] ([yshift=-0.5cm]T1.east) .. controls ([xshift=0.4cm]T1.40) and (T1.350) .. (T1.center) node[] (N1) {};
        \draw[black,dashed] (T1.center) .. controls ([xshift=2.4cm]T1.110) .. ([xshift=-0.3cm]T1.west) node[] (N2) {{}};
        
        \draw[-{Stealth[scale=1]}] ([xshift=0.35cm,yshift=0.5cm]T1.west) to ([xshift=-0.3cm]T1.west);
        \draw[] ([xshift=0.4cm,yshift=0.7cm]T1.west) node [right,black, text width=2.2cm] {\footnotesize Leaf $\ell$ with $\mathrm{Dens}_H(\ell)\approx\delta k$};

        \draw [white,decorate,decoration={brace,mirror,raise=1pt,amplitude=3pt}] ([xshift=0.2cm,yshift=-0.5cm]T1.left corner) -- ([xshift=2.6cm,yshift=-0.5cm]T1.left corner) node [white,pos=0.5,yshift=-0.65cm] {{}};
        
        \draw [] ([xshift=0cm]T1.left corner) -- (T1.right corner) node [black,pos=0.5,below] { $\approx\delta k~~\approx\delta k~~\approx\delta k~~\approx\delta k~~ \approx\delta k~~ \approx\delta k$\vspace{2em}};
        \end{tikzpicture}
  \caption{An illustration of the good case where $\mathrm{Dens}_H$ is concentrated around its mean of $\delta k$. }
  \label{fig:good case}
\end{subfigure}
\hspace*{\fill}
  \captionsetup{width=.9\linewidth}
\caption{An illustration of our resilience lemma (\Cref{lem:resilience lemma}). This lemma shows that all trees resemble the one on the right, with $\mathrm{Dens}_H(\bell)$ tightly concentrated around its mean of $\delta k$. This allows us to rule out bad trees such as those on the left where all of the hardness is concentrated on a small fraction of the leaves.}
\label{fig:resilience lemma}
\end{figure}
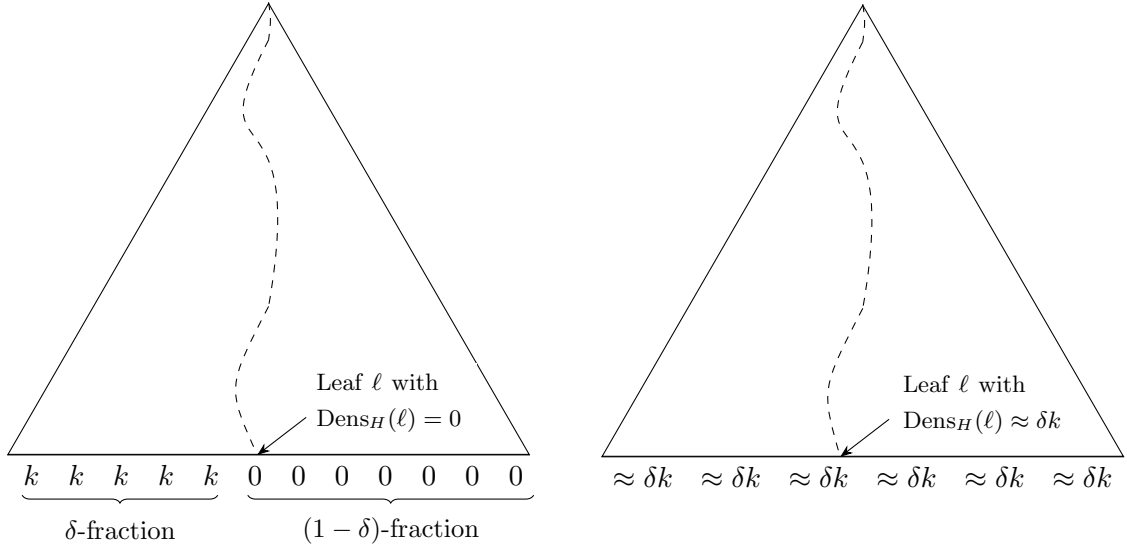

Comparing~\Cref{lem:resilience lemma} to~\Cref{claim:expected total hardcore advantage}, we see that~\Cref{claim:expected total hardcore advantage} is a statement about density preservation {\sl in expectation} whereas~\Cref{lem:resilience lemma} is a statement about density preservation {\sl with high probability}. It says that $H$'s density at a random leaf $\bell\sim \mu^k(T)$ and in a random block $\bi\sim [k]$ remains, with high probability, roughly the same as that of $H$'s initial density---this is why we call~\Cref{lem:resilience lemma} a resilience lemma. (Our proof of~\Cref{lem:resilience lemma} in fact shows that the $H$'s density remains resilient under arbitrary partitions of $(\{ \pm 1\}^n)^k$, not just those induced by a decision tree.)

With~\Cref{lem:resilience lemma} in hand, the intuition sketched in~\Cref{sec:jensen} can be made formal. 

\section{Discussion and Future Work}

Our main results are a strong direct sum theorem and a strong XOR lemma for distributional query complexity, showing that if $f$ is somewhat hard to approximate with depth-$d$ decision trees, then $f^{\otimes k}$ and $f^{\oplus k}$ are both {\sl much harder} to approximate, even with decision trees of {\sl much larger} depth.  These results hold for expected query complexity, and they circumvent a counterexample of Shaltiel showing that such statements are badly false for worst-case query complexity.  We view our work as confirming a remark Shaltiel made in his paper, that his counterexample ``seems to exploit defects in the formulation of the problem rather than show that our general intuition for direct product assertions is false."

Shaltiel's counterexample applies to  
 many other models including boolean circuits and communication protocols. A broad avenue for future work is to understand how this counterexample can be similarly circumvented in these models by working with more fine-grained notions of computation cost.  Consider for example boolean circuit complexity and Yao's XOR lemma~\cite{Yao82}, which states that if $f$ is mildly hard to approximate with size-$s$ circuits w.r.t.~$\mu$, then $f^{\oplus k}$ is extremely hard to approximate with size-$s'$ circuits w.r.t.~$\mu^{k}$.  A well-known downside of this important result is that it only holds for $s'\ll s$. Indeed, Shaltiel's counterexample shows that it cannot hold for $s' \gg s$, at least not for the standard notion of circuit size. Extrapolating from our work, can we prove a strong XOR lemma for boolean circuits by considering a notion of the ``expected size" of a circuit $C : \bn \to \bits$ with respect to a distribution $\mu$ over $\bn$?  A natural approach is to consider the standard notion of the expected runtime of a Turing machine with respect to a distribution over inputs and have the Cook--Levin theorem guide us towards an appropriate analogue for circuit size.

On a more technical level, a crucial ingredient in our work is the first use of Impagliazzo's Hardcore Theorem within the context of query complexity (and indeed, to our knowledge, the first use of it outside of circuit complexity). Could this powerful theorem be useful for other problems in query complexity, possibly when used in conjunction with our new resilience lemma?

\section{Preliminaries}

We use $[n]$ to denote the set $\set{1,2,\ldots, n}$ and \textbf{bold font} (e.g $\bx \sim \mcD$) to denote random variables. For any distribution $\mu$, we use $\mu^k$ to denote $k$-fold the product distribution $\mu \times \cdots \times \mu$.

For any function $f:\bn \to \bits$, we use $f^{\otimes k}: (\bn)^k \to \bits^k$ to denote its $k$-fold direct product,
\begin{equation*}
    f^{\otimes k}(X^{(1)}, \ldots, X^{(k)}) \coloneqq \paren[\big]{f(X^{(1)}), \ldots, f(X^{(k)})}.
\end{equation*}
Similarly, we use $f^{\oplus k}: (\bn)^k \to \bits$ to denote its $k$-fold direct sum,
\begin{equation*}
    f^{\oplus k}(X^{(1)}, \ldots, X^{(k)}) \coloneqq \prod_{i \in [k]}f(X^{(i)}).
\end{equation*}

\begin{definition}[Bernoulli distribution]
    For any $\delta \in [0,1]$, we write $\Ber(\delta)$ to denote the distribution of $\bz$ where $\bz = 1$ with probability $\delta$ and $0$ otherwise.
\end{definition}

\begin{definition}[Binomial distribution]
    For any $k \in \N$, $\delta \in [0,1]$, we write $\Bin(k, \delta)$ to denote the sum of $k$ independent random variables drawn from $\Ber(\delta)$.
\end{definition}

\begin{fact}[Chernoff bound]
    \label{fact:chernoff}
    Let $\bz_1, \ldots, \bz_k$ be independent and each bounded within $[0,1]$ and $\bZ \coloneqq \sum_{i \in [k]}\bz_i$. For any threshold $t \leq \mu \coloneqq \Ex[\bZ]$,
    \begin{equation*}
        \Pr[\bZ \leq t] \leq \exp\paren*{-\tfrac{(\mu - t)^2}{2\mu}}.
    \end{equation*}
    Similar bounds for the probability $\bZ$ exceeds its mean hold. For example,
    \begin{equation*}
        \Pr[\bZ \geq 2\mu] \leq \exp\paren*{-\tfrac{\mu}{3}}.
    \end{equation*}
    Furthermore, the above bounds also hold for any random variable $\bY$ satisfying $\Ex[e^{\lambda \bY}] \leq \Ex[e^{\lambda \bZ}]$ for all $\lambda \in \R$.
\end{fact}

\subsection{Randomized vs.~deterministic decision trees}
We will prove all of our results with respect to the expected depth of a \emph{randomized} decision tree. In this subsection, we formally define \emph{deterministic} and \emph{randomized} decision trees and prove that our results easily extend to the deterministic setting.

\begin{definition}[Deterministic decision tree]
    A deterministic decision tree, $T: \bits^n \to \bits$, is a binary tree with two types of nodes: Internal nodes each query some $x_i$ for $i \in [n]$ and have two children whereas leaf nodes are labeled by a bit $b \in \bits$ and have no children. On input $x \in \bits^n$, $T(x)$ is computed as follows: We proceed through $T$ starting at the root. Whenever at an internal node that queries the $i^{\text{th}}$ coordinate, we proceed to the left child if $x_i = -1$ and right child if $x_i = +1$. Once we reach a leaf, we output the label of that leaf.
\end{definition}
\begin{definition}[Randomized decision tree]
    A randomized decision tree, $\mcT:\bits^n \to \bits$, is distribution over deterministic decision trees. On input $x \in \bits^n$, it first draws $\bT \sim \mcT$ and then outputs $\bT(x)$.
\end{definition}
\begin{definition}[Expected depth]
    For any deterministic decision tree $T: \bits^n \to \bits$ and distribution $\mu$ on $\bits^n$, we use $\overline{\mathrm{Depth}}^{\mu}(T)$ to denote the expected depth of $T$, which is the expected number of coordinates $T$ queries on a random input $\bx \sim \mu$. Similarly, for a randomized decision tree $\mcT$, $\overline{\mathrm{Depth}}^{\mu}(\mcT) \coloneqq \Ex_{\bT \sim \mcT}[\overline{\mathrm{Depth}}^{\mu}(\bT)]$.
\end{definition}

We write $\overline{\mathrm{Depth}}^{\mu}(\cdot,\cdot)$ to denote the minimum expected depth of \emph{any} decision tree, including randomized decision trees. Thus, all of our main results, as written, hold for \emph{randomized} decision trees; however, equivalent statements are true if we restrict ourselves to only \emph{deterministic} decision trees, with only a small change in constants, as easily seen from the following claim.
\begin{claim}
    \label{claim:connect-det-rand}
    For any $f:\bn \to \bits$, distribution $\mu$, and constant $\eps$,
    \begin{equation*}
         \odepthmu(f,2\eps) \leq \odepthmu_{\det}(f, 2\eps) \leq 2 \cdot \odepthmu(f, \eps).
    \end{equation*}
\end{claim}
\begin{proof}
    The left-most inequality follows immediately from that fact that any deterministic decision tree is also a randomized decision tree. For the second inequality, given any randomized decision tree $\mcT$ with error $\eps$ and expected depth $d$, we'll construct a deterministic decision tree $T$ with error at most $2\eps$ and expected depth at most $2d$. First, we decompose the expected error of $\mcT$:
    \begin{equation*}
        \eps = \Prx_{\bx \sim \mu}[\mcT(\bx) \neq f(\bx)] = \Ex_{\bT \sim \mcT}\bracket[\Big]{\Prx_{\bx \sim \mu}[\bT(\bx) \neq f(\bx)]}.
    \end{equation*}
    Applying Markov's inequality, if we sample $\bT \sim \mcT$, with probability at least $1/2$, it has error at most $\eps$. Similarly, since the expected depth of $\mcT$ is $d$, with probability at least $1/2$, $\bT$ will have expected depth at most $2d$. By union bound, there is a nonzero probability that we choose a single (deterministic) tree with error at most $2\eps$ and expected depth at most $2d$.
\end{proof}
\Cref{claim:connect-det-rand} immediately allows direct sum theorems for \emph{randomized} decision trees to also apply to \emph{deterministic} decision trees.
\begin{corollary}[Randomized direct sum theorems imply deterministic ones]
    \label{cor:rand-to-det-DP}
    Suppose that a randomized direct sum theorem of the following form holds. For a function $f: \bn \to \bits$, distribution $\mu$  over $\bn$,  $k\in \mathbb{N}$, and constants $\varepsilon, \delta, M$,
    \[ \overline{\mathrm{Depth}}^{\mu^{\otimes k}}(f^{\otimes k},\varepsilon) \ge M \cdot \overline{\mathrm{Depth}}^{\mu}(f,\delta).\]
    Then, 
     \[ \overline{\mathrm{Depth}}_{\det}^{\mu^{\otimes k}}(f^{\otimes k},\varepsilon) \geq \overline{\mathrm{Depth}}^{\mu^{\otimes k}}(f^{\otimes k},\varepsilon) \ge M \cdot \overline{\mathrm{Depth}}^{\mu}(f,\delta) \geq  \tfrac{M}2 \cdot \overline{\mathrm{Depth}}^{\mu}_{\mathrm{det}}(f,2\delta).\]
\end{corollary}
With \Cref{cor:rand-to-det-DP} in mind, the remainder of this paper will only consider \emph{randomized} decision trees.

\section{Proof of~\Cref{thm:main-formal}}
The purpose of this section is to prove our direct sum theorem (\Cref{thm:main-formal}) showing that if $f$ is hard to compute, than $f^{\otimes k}$ is even harder to compute. We will in fact prove a {\sl threshold} direct sum theorem, showing that it is hard even to get {\sl most} of the $k$ copies correct. To formalize this, we generalize the notation $\odepth(\cdot, \cdot)$ to take in an additional threshold parameter $t$ specifying how many blocks we allow to be wrong.
\begin{definition}
    \label{def:threshold-DPT}
    For any function $f:\bn \to \bits$, error $\eps$, and threshold $t \in \N$, we use $\odepthmuk(f^{\otimes}, \eps, t)$ to denote the minimum expected depth of a tree $T:(\bn)^k \to \bits^k$ satisfying
    \begin{equation*}
        \Prx_{\bX \sim \mu^k}\bracket*{\lzero{T(\bX) - f^{\otimes k}(\bX)}] > t} \leq \eps.
    \end{equation*}
\end{definition}

\begin{theorem}[A strong threshold direct sum theorem for query complexity]
    \label{thm:threshold-DPT}
    For every function $f:\bits^n \to \bits$, distribution $\mu$ over $\bits^n$, $k \in \N$, and $\gamma, \delta \in (0,1)$,
    \begin{equation*}
        \odepthmuk(f^{\otimes k}, 1 - e^{-\Omega(\delta k)} - \gamma, \Omega(\delta k)) \geq \Omega\left( \frac{\gamma^2 k}{\log(1/\delta)}\right) \cdot \odepthmu(f,\delta).
    \end{equation*}
\end{theorem}
Note that the threshold of $\Omega(\delta k)$ is within a constant factor of optimal, as repeating an algorithm that errs $\delta$ fraction of the time $k$ times will lead to an average of $\delta k$ mistakes.~\Cref{thm:threshold-DPT} implies our standard strong direct sum theorem (\Cref{thm:main-formal}) because \[ \odepthmuk(f^{\otimes k}, \eps, t) \geq \odepthmuk(f^{\otimes k}, \eps, 0) = \odepthmuk(f^{\otimes k}, \eps)\]  for any $t \geq 0$ and $\eps > 0$.

\subsection{The structure of this section}

By the Hardcore Theorem (\Cref{thm:hardcore theorem for DTs}), proving~\Cref{thm:threshold-DPT} reduces to proving:  \medskip 

\begin{tcolorbox}[colback = white,arc=1mm, boxrule=0.25mm]
\begin{theorem}[Hardness of $f^{\otimes k}$ in terms of a hardcore measure for $f$]
    \label{thm:bounds-from-hardcore}
    Suppose that $f:\bn \to \bits$ has an $(\gamma, d)$-hardcore measure w.r.t $\mu$ of density $\delta$. Then, for any $T:(\bn)^k \to \bits^k$ with $\odepthmuk(T) \leq kd$,
    \begin{equation*}
        \Prx_{\bX \sim \mu^k}\bracket*{\lzero{T(\bX) - f^{\otimes k}(\bX)} \leq \lfrac{\delta k}{10}} \leq e^{-\frac{\delta k}{10}} + 10\gamma.
    \end{equation*}
\end{theorem}
\end{tcolorbox}
\medskip 

This section is therefore devoted to proving \Cref{thm:bounds-from-hardcore}.  As discussed in \Cref{sec:tech-overview}, our proof tracks two key quantities: we will analyze how {\sl hardcore density} (\Cref{def:hardcore-density}) and {\sl hardcore advantage} (\Cref{def:hardcore advantage}) are distributed over the leaves of $T$. This proof will be broken into three steps:
\begin{enumerate}
    \item In \Cref{subsec:understand-hardcore}, we prove \Cref{claim:expected total hardcore advantage} and \Cref{lem:resilience lemma}, which aim to understand the distributions of the hardcore density and hardcore advantage of a random leaf of $T$.
    \item In \Cref{subsec:error-one-leaf}, we derive an expression for the probability $T$ makes surprisingly few mistakes as a function of the hardcore density and hardcore advantage at each leaf. This generalizes \Cref{lem:accuracy in terms of hardcore density and advantage}.
    \item In \Cref{subsec:threshold-DPT-proof}, we combine the above  to prove \Cref{thm:bounds-from-hardcore}.
\end{enumerate}

\subsection{How hardcore density and advantage distribute over the leaves}
\label{subsec:understand-hardcore}
We begin with proving our resilience lemma for hardcore density. Roughly speaking, this will say that for any tree $T:(\bn)^k \to \bits^k$, the hardcore density of a random leaf concentrates around $\delta k$. We recall the definition of hardcore density.
\begin{definition}[Hardcore density at $\ell$, \Cref{def:hardcore-density} restated]
\label{def:hardcore-density-body}
For any tree $T:(\bn)^k \to \bits^k$, hardcore measure $H:\bn \to [0,1]$, distribution $\mu$ on $\bn$, $i\in [k]$, and leaf $\ell$ of $T$, the {\sl hardcore density at $\ell$ in the $i$th block} is the quantity: 
\[ \Dens_H(\ell,i)\coloneqq \Ex_{\bX\sim \mu^k}\big[H(\bX^{(i)})\mid \text{$\bX$ reaches $\ell$}\big]. \]
The {\sl total hardcore density at $\ell$} is the quantity $\ds \Dens_H(\ell) \coloneqq \sum_{i =1}^k \Dens_H(\ell,i).$
\end{definition}
The distribution over leaves in the resilience lemma is the canonical distribution.
\begin{definition}[Canonical distribution]
    \label{def:canonical-dist}
    For any tree $T$ and distribution $\mu$ over $T$'s domain, we write $\mu(T)$ to denote the distribution over leaves of $T$ where:  
\[ \Prx_{\bell\sim \mu^k(T)}[\bell = \ell] = \Prx_{\bX\sim \mu}[\text{$\bX$ reaches $\ell$}].\]
\end{definition}
\begin{lemma}[Resilience lemma, generalization of \Cref{lem:resilience lemma}]
    \label{lem:resilience-general}
    For any $T:(\bits^n)^k \to \bits^k$, hardcore measure $H$, distribution $\mu$, and convex $\Phi:\R \to \R$,
    \begin{equation*}
        \Ex_{\bell \sim \mu^k(T)}[\Phi(\dens_H(\bell))] \leq \Ex_{\bz \sim \Bin(k,\delta)}[\Phi(\bz)]
    \end{equation*}
    where $\delta \coloneqq \Ex_{\bx \sim \mu}[H(\bx)]$ is the density of $H$ w.r.t. $\mu$.
\end{lemma}
\begin{proof}
    Draw $\bX \sim \mu^k$. Then, for each $i \in [k]$, independently draw $\bz_i \sim \Ber(H(\bX^{(i)}))$. Note that, for any leaf $\ell$ and $i \in [k]$,
    \begin{equation*}
        \dens_H(\ell, i) \coloneqq \Ex_{\bX \sim \mu^k}[H(\bX^{(i)}) \mid \bX\text{ reaches }\bell] = \Ex_{\bX \sim \mu^k}[\bz_i \mid \bX\text{ reaches }\bell].
    \end{equation*}
    By the above equality and definition $\dens_H(\ell) = \sum_{i \in [k]}\dens_H(\ell,i)$,
    \begin{align*}
        \Ex_{\bell \sim \mu^k(T)}[\Phi(\dens_H(\bell))] &= \Ex_{\bell \sim \mu^k(T)}\bracket[\Bigg]{\Phi\paren[\Bigg]{ \Ex_{\bX \sim \mu^k}\bracket[\Bigg]{\sum_{i \in [k]}\bz_i \mid \bX\text{ reaches }\bell}}}\\
        &\leq \Ex_{\bell \sim \mu^k(T)}\bracket[\Bigg]{\Ex_{\bX \sim \mu^k}\bracket[\Bigg]{\Phi\paren[\Bigg]{ \sum_{i \in [k]}\bz_i \mid \bX\text{ reaches }\bell}}} \tag{Jensen's inequality} \\
        &= \Ex_{\bX \sim \mu^k}\bracket[\Bigg]{\Phi\paren[\Bigg]{ \sum_{i \in [k]}\bz_i}}. \tag{Law of total expectation}
    \end{align*}
    Note that the last line holds precisely for the distribution $\mu^k(T)$ defined in \Cref{def:canonical-dist}, which is why we use that distribution.

    Since $\bX$ is drawn from a product distribution, and $\bz_i$ depends on only the $i^{\text{th}}$ coordinate of $\bX$, $\bz_1,\ldots, \bz_k$ are independent. Furthermore, each has mean $\Ex_{\bx \sim \mu}[H(\bx)] = \delta$. Therefore, $\sum_{i \in [k]}\bz_i$ is distributed according to $\Bin(k,\delta)$.
\end{proof}
A couple of remarks about the above Lemma: First, it implies that $\dens_H(\bell)$ concentrates around $\delta k$. Since $z \mapsto e^{\lambda z}$ is convex for any $\lambda \in \R$,  \Cref{lem:resilience-general} implies that the moment generating function of $\dens_H(\bell)$ is dominated by that of $\Bin(k,\delta)$. This means that Chernoff bounds that hold for $\Bin(k,\delta)$ also hold for $\dens_H(\bell)$. In particular, the statement of \Cref{lem:resilience lemma} is a consequence of the Chernoff bound given in \Cref{fact:chernoff}.

Second, the proof of \Cref{lem:resilience-general} does not make heavy use of the decision tree structure of $T$. It only uses that the leaves of $T$ partition $(\bn)^k$, and so may find uses for other models that partition the domain.

\paragraph{Depth amplification for hardcore advantage.}

While the resilience lemma gives a fairly fine-grained understanding of how hardcore density distributes among the leaves, our guarantee for hardcore advantage are more coarse -- that its expectation over the leaves is bounded.

\begin{definition}[Hardcore advantage at $\ell$, \Cref{def:hardcore advantage} restated]
\label{def:hardcore-advantage-body}
For any tree $T:(\bn)^k \to \bits^k$, hardcore measure $H:\bn \to [0,1]$, distribution $\mu$ on $\bn$, $i\in [k]$, and leaf $\ell$ of $T$, the {\sl hardcore advantage at $\ell$ in the $i$th block} is the quantity: 
\begin{equation}
    \label{eq:hardcore-advantage}
     \Adv_H(\ell,i) \coloneqq \abs[\Big]{\Ex_{\bX\sim \mu^k}\big[ f(\bX^{(i)})T(\bX)_i H(\bX^{(i)})\mid \text{$\bX$ reaches $\ell$}\big]}.
\end{equation}
The {\sl total hardcore advantage at $\ell$} is the quantity  
$ \ds \Adv_H(\ell) \coloneqq \sum_{i=1}^k \Adv_H(\ell, i). $
\end{definition}

\begin{lemma}[Expected total hardcore advantage, \Cref{claim:expected total hardcore advantage} restated]
\label{lem:expected total hardcore advantage body}
    If $H$ is a $(\gamma, d)$-hardcore measure for $f$ of density $\delta$ w.r.t.~$\mu$ and the expected depth of $T$ is at most $dk$, then 
    \begin{equation*}
        \Ex_{\bell \sim \mu^k(T)}[\adv_H(\bell)] \leq \gamma \Ex_{\bell\sim\mu^k(T)}[\Dens_H(\bell)] = \gamma \delta k. 
    \end{equation*}
\end{lemma}
\begin{proof}[Proof of \Cref{lem:expected total hardcore advantage body}]
    By contrapositive. Suppose there exists $T_{\textnormal{large}}:(\bn)^k \to \bits^k$ making $dk$ queries on average w.r.t. $\mu^k$ for which,
    \begin{equation*}
        \Ex_{\bell \sim \mu^k(T)}[\adv_H(\bell)] > \gamma \Ex_{\bell\sim \mu^k(T)}[\dens_H(\bell)] = \gamma\delta k.
    \end{equation*}
    Then, we'll show there exists $T_{\textnormal{small}}:\bn \to \bits$ making $d$ queries on average w.r.t. $\mu$ for which
    \begin{equation*}
        \Ex_{\bx \sim \mu}[ f(\bx) T_{\textnormal{small}}(\bx) H(\bx)] > \gamma \cdot \Ex_{\bx \sim \mu}[H(\bx)] = \gamma \delta.
    \end{equation*}
    Before constructing $T_{\textnormal{small}}$, we observe that we can assume, without loss of generality, that for every leaf $\ell$ of $T$, that we can remove the absolute value from \Cref{eq:hardcore-advantage}; i.e. that
    \begin{equation*}
        \Adv_H(\ell,i) = \Ex_{\bX\sim \mu^k}\big[ f(\bX^{(i)})T(\bX)_i H(\bX^{(i)})\mid \text{$\bX$ reaches $\ell$}\big]
    \end{equation*}
    Otherwise, we could modify this leaf by flipping the label of $T(X)_i$ whenever $X$ reaches a leaf where the above quantity is negative. This does not change the hardcore advantage, so this new $T$ still satisfies our assumption.
    
    $T_{\textnormal{small}}$ will be a randomized algorithm. Upon receiving the input $x \in \bn$, it samples $\bX \sim \mu^k$ and $\bi \sim \mathrm{Unif}([k])$, and then constructs $\bX(x, \bi)$ by inserting $x$ into the $\bi^{\text{th}}$ block of $\bX$,
    \begin{equation*}
        (\bX(x, i))^{(j)} = \begin{cases}
            \bX^{(j)} &\text{if }j \neq i\\
            x&\text{if }j = \bi.
        \end{cases}
    \end{equation*}
    Then, $T_{\textnormal{small}}(x)$ outputs $T_{\textnormal{large}}(\bX(x, \bi))_{\bi}$.

    Our analysis of $T_{\textnormal{small}}$ relies on the following simple observation: If we sample $\bx \sim \mu$, then even conditioning on any choice of $\bi = i$, the distribution of $\bX(\bx, \bi)$ is $\mu^k$. This also means that $\bX(\bx, \bi)$ and $\bi$ are independent.

    We claim that $T_{\textnormal{small}}$ has the two desired properties; low expected number of queries, and high accuracy on $H$. To bound the expected number of queries $T_{\textnormal{small}}$ makes on an input $\bx \sim \mu$, we use that $\bX(\bx, \bi)$ is distributed according to $\mu^k$. Therefore, $T_{\textnormal{large}}(\bX(\bx, \bi))$ makes, on average, $dk$ queries. Expanded, we have that,
    \begin{equation*}
        \sum_{i \in [k], j\in [n]}\Pr[T_{\textnormal{large}}(\bX(\bx, \bi))\text{ queries }\bX(\bx)^{(i)}_j] = dk.
    \end{equation*}
    Whereas, the number of queries $T_{\textnormal{small}}(\bx)$ makes only counts queries to the $\bi^{\text{th}}$ block, and is therefore,
    \begin{align*}
         \sum_{i \in [k], j\in [n]}\Pr\big[T_{\textnormal{large}}(\bX(\bx, \bi))&\text{ queries }\bX(\bx, \bi)^{(i)}_j \cdot \Ind[\bi = i]\big] \\
         &= \sum_{i \in [k], j\in [n]}\Pr\big[T_{\textnormal{large}}(\bX(\bx, \bi)))\text{ queries }\bX(\bx, \bi)^{(i)}_j\big] \cdot \frac{1}{k}\\
         &= d.
    \end{align*}
    In the above, the first equality uses that $\bi$ is independent of $\bX(\bx, \bi)$, and so is still uniform on $[k]$ even conditioned on which queries $T_{\textnormal{large}}$ makes.

    Lastly, we verify that $T_{\textnormal{small}}$ has high accuracy on the hardcore measure.
    \begin{align*}
        \Ex_{\bx \sim \mu}[ f(\bx) T_{\textnormal{small}}(\bx, \bi)) H(\bx)] &= \Ex_{\bx \sim \mu}[ f(\bx) T_{\textnormal{large}}(\bX(\bx, \bi)))_{\bi} H(\bx)] \tag{Definition of $T_{\textnormal{small}}$}  \\
        &=\Ex_{\bi \sim [k]}\bracket*{\Ex_{\bX \sim \mu^k}\bracket*{f(\bX^{(\bi)}) T_{\textnormal{large}}(\bX)_{\bi} H(\bX^{(\bi)})}} \tag{$\bi, \bX(\bx)$ are independent} \\
        &= \Ex_{\bi \sim [k]}\bracket*{\Ex_{\bell \sim \mu^k(T_{\textnormal{large}})}[\adv_H(\bell, \bi)]} \tag{\Cref{def:hardcore-advantage-body}}\\
        &= \frac{1}{k}\Ex_{\bell \sim \mu^k(T_{\textnormal{large}})}[\adv_H(\bell)] > \gamma \delta. \tag*{\qedhere}
    \end{align*}
\end{proof}

\subsection{Understanding the error in terms of hardcore density and advantage}
\label{subsec:error-one-leaf}
To state the main result of this subsection, we'll define the following distribution for the sum of independent Bernoulli random variables.
\begin{definition}
    For any $p \in [0,1]^k$, we use $\BerSum(p)$ to denote the distribution of $\bz \coloneqq \bz_1 + \cdots + \bz_k$ where each $\bz_i$ is independently drawn from $\Ber(p_i)$.
\end{definition}
The following generalizes \Cref{lem:accuracy in terms of hardcore density and advantage}.
\begin{lemma}[Accuracy in terms of hardcore density and advantage of the leaves]
    \label{lem:threshold-accuracy-from-stats}
    Let $H$ be a $(\gamma, d)$-hardcore measure w.r.t.~$\mu$ for $f: \bits^n \to \bits$, and $T:(\bn)^k \to \bits$ be any tree. Then, for any $t \geq 0$, 
    \begin{equation*}
        \Prx_{\bX \sim \mu^k}\bracket*{\lzero[\big]{T(\bX)-  f^{\otimes k}(\bX)} \leq t} \leq \Ex_{\bell \sim \mu^k(T)}\bracket*{\Prx_{\bz \sim \BerSum(p(\bell))}[\bz \leq t]}
    \end{equation*}
    where $p(\ell) \in [0,1]^k$ is the vector where
    \begin{equation*}
        p(\ell)_i \coloneqq \frac{\dens_H(\ell,i) - \adv_H(\ell,i)}{2}\quad\quad\text{for each }i \in [k].
    \end{equation*}
\end{lemma}
\Cref{lem:threshold-accuracy-from-stats} implies a generalization of \Cref{lem:accuracy in terms of hardcore density and advantage}.
\begin{corollary}
    \label{cor:threshold-from-stats-chernoff}
    Let $H$ be a $(\gamma, d)$-hardcore measure w.r.t.~$\mu$ for $f: \bits^n \to \bits$, and $T:(\bn)^k \to \bits$ be any tree. Then, for any $t \geq 0$,
    \begin{equation*}
        \Prx_{\bX \sim \mu^k}\bracket[\Big]{\lzero[\big]{T(\bX)- f^{\otimes k}(\bX)} \leq t} \leq \Ex_{\bell \sim \mu^k(T)}\bracket*{\min\paren[\Big]{1, \exp\paren[\Big]{t - \frac{\Dens_H(\bell) - \Adv_H(\bell)}{4}}}}.
    \end{equation*}
\end{corollary}
\begin{proof}
    The Chernoff bound of \Cref{fact:chernoff} says that, for any $p \in [0,1]^k$ and $\mu \coloneqq \sum_{i \in [k]} p_i$,
    \begin{equation*}
        \Prx_{\bz \sim \BerSum(p)}[\bz \leq t] \leq \begin{cases}
            \exp\paren*{-\frac{(\mu - t)^2}{2\mu}}&\text{if }\mu \geq t \\
            1& \text{otherwise.}
        \end{cases}
    \end{equation*}
    We want to show that the above is bounded by $\min(1, e^{t- \mu/2})$. Clearly this holds for $\mu < t$, so we need only consider the case where $\mu \geq t$
    \begin{equation*}
         \exp\paren*{-\tfrac{(\mu - t)^2}{2\mu}} = \exp\paren*{-\tfrac{\mu^2 - 2t\mu + t^2}{2\mu}} \leq \exp\paren*{-\tfrac{\mu^2 - 2t\mu}{2\mu}} = e^{t - \mu/2}.
    \end{equation*}
    Since $ \exp\paren*{-\frac{(\mu - t)^2}{2\mu}} \leq 1$ as well, it is upper bounded by $\min(1, e^{t- \mu/2})$ as desired. The desired result follows from \Cref{lem:threshold-accuracy-from-stats} as well as $\sum_{i \in [k]} p(\ell)_i = \frac{\Dens_H(\ell) - \Adv_H(\ell)}{2}$ for every leaf $\ell$ of $T$.
\end{proof}

The main observation underlying \Cref{lem:threshold-accuracy-from-stats} is that, if we choose an input $\bX \sim \mu^k$ conditioned on reaching a leaf $\ell \in T$, that $\bX$ is distributed according to a $k$-wise product distribution (i.e.~from $\mu_1(\ell) \times \cdots \times \mu_k(\ell)$ for appropriately defined distributions). The below is essentially the same as  Lemma 3.2 of \cite{Dru12}, but we include a proof for completeness.
\begin{claim}
    \label{claim:leaf-product}
    For any (potentially randomized) tree $T:(\bn)^k \to \mcY$ and leaf $\ell$ of $T$, if $\bX \sim \mu^k$, then the distribution of $\bX$ conditioned on reaching the leaf $\ell$ is a product distribution over the $k$ blocks of $\bX$.
\end{claim}
\begin{proof}
    First, if $T$ is a randomized tree, it is as a distribution over deterministic trees. If the desired result holds for each of those deterministic trees, it also holds for $T$. Therefore, it suffices to consider the case where $T$ is deterministic.

    We'll prove that the distribution of $\bX$ reaching any internal node \emph{or} leaf of $T$ is product by induction on the depth of that node. If that depth is $0$, then all inputs reach it and so the desired result follows from $\mu^k$ being product.

    For depth $d \geq 1$, let $\alpha$ be the parent of $\ell$. Then $\alpha$ has depth $d-1$, so by the inductive hypothesis, the distribution of inputs reaching $\alpha$ is product. Let $i \in [k], j \in [n], b \in \bits$ be chosen so that an input $X$ reaches $\ell$ iff it reaches $\alpha$ and $X^{(i)}_j = b$. Then,
    \begin{align*}
        \Pr[&\bX =X \mid \bX \text{ reaches } \ell] \\
        &= \Pr[\bX =X\mid \bX \text{ reaches } \alpha] \cdot \frac{\Ind[X^{(i)}_j = b]}{\Pr[X^{(i)}_j=b \mid \bX \text{ reaches } \alpha]} \\
         &= \frac{\Ind[X^{(i)}_j = b]}{\Pr[X^{(i)}_j =b\mid \bX \text{ reaches } \alpha]} \cdot \prod_{\ell \in [k]}  \Pr[\bX^{(\ell)} =X^{(\ell)}\mid \bX \text{ reaches } \alpha] \tag{Inductive hypothesis}\\
         &= \paren*{\prod_{\ell \neq i}  \Pr[\bX^{(\ell)} =X^{(\ell)}\mid \bX \text{ reaches } \alpha]} \cdot \frac{\Pr[\bX^{(i)} =X^{(i)}\mid \bX \text{ reaches } \alpha] \cdot \Ind[X^{(i)}_j = b]}{\Pr[X^{(i)}_j =b\mid \bX \text{ reaches } \alpha]}\\
         &= \paren*{\prod_{\ell \neq i}  \Pr[\bX^{(\ell)} =X^{(\ell)}\mid \bX \text{ reaches } \alpha]} \cdot \Pr[\bX^{(i)} =X^{(i)}\mid \bX \text{ reaches } \alpha, \bX^{(i)}_j = b].
    \end{align*}
    The above is decomposed as a product over the $k$ components of $X$, so is a product distribution.
\end{proof}
We conclude this subsection with a proof of \Cref{lem:threshold-accuracy-from-stats}.
\begin{proof}[Proof of \Cref{lem:threshold-accuracy-from-stats}]
    Consider any leaf $\ell$ of $T$. We wish to compute the probability that $T(\bX)$ makes less than $t$ mistakes on $f^{\otimes k}(\bX)$ given that $\bX$ reaches the leaf $\ell$. On this leaf, $T$ outputs a single vector $y \in \bits^k$. Meanwhile, by \Cref{claim:leaf-product}, the distribution of $\bX$ is product over the blocks, and so $f(\bX^{(1)}), \ldots, f(\bX^{(k)})$ are independent. Define $q(\ell) \in [0,1]^k$ as,
    \begin{equation*}
        q(\ell)_i \coloneqq \Prx_{\bX \sim \mu^k}[y_i \neq f(\bX^{(i)})].
    \end{equation*}
    Then,
    \begin{equation*}
        \Prx_{\bX \sim \mu^k}\bracket[\big]{\lzero[\big]{T(\bX), f^{\otimes k}(\bX)} \leq t \mid \bX\text{ reaches }\ell} = \Prx_{\bz \sim \BerSum(q(\ell))}[\bz \leq t].
    \end{equation*}
    For $\bz \sim \BerSum(q)$, the probability $\bz \leq t$ is monotonically decreasing in each $q_i$. Therefore, it suffices to show that $q(\ell)_i \geq p(\ell)_i$ for each $i \in [k]$. We compute,
    \begin{align*}
        q(\ell)_i &= \Prx_{\bX \sim \mu^k}[T(\bX)_i \neq f(\bX^{(i)}) \mid \bX\text{ reaches }\ell] \\
        &=  \frac{1 - \Ex_{\bX \sim \mu^k}[T(\bX)_i f(\bX^{(i)}) \mid \bX\text{ reaches }\ell]}{2} \tag{$T(X)_i, f(X^{(i)} \in \bits$}
    \end{align*}
    Separating the above expectation into two pieces, for $\bX$ drawn from $\mu^k$ conditioned in $\bX$ reaching $\ell$,
    \begin{align*}
        \Ex[T(\bX)_i f(\bX^{(i)})] &= \Ex[T(\bX)_i f(\bX^{(i)})H(\bX^{(i)})] +  \Ex[T(\bX)_i f(\bX^{(i)})(1-H(\bX^{(i)}))] \\
        &\leq \adv_H(\ell,i) +  \Ex[T(\bX)_i f(\bX^{(i)})(1-H(\bX^{(i)}))] \tag{\Cref{def:hardcore advantage}}\\
        &\leq \adv_H(\ell,i) +  \Ex[(1-H(\bX^{(i)}))] = 1 +\adv_H(\ell,i) - \dens_H(\ell,i).
    \end{align*}
    Therefore,
    \begin{equation*}
        q(\ell)_i \geq \frac{\dens_H(\ell,i)- \adv_H(\ell,i)}{2} = p(\ell)_i. \qedhere
    \end{equation*}
\end{proof}

\subsection{Completing the proof of the threshold direct sum theorem}
\label{subsec:threshold-DPT-proof}
In this subsection, we complete the proof of \Cref{thm:bounds-from-hardcore}. Throughout this section, we'll use the following function:
\begin{equation*}
    g_t(z) \coloneqq \min(1, e^{t - z/4}).
\end{equation*}
By using \Cref{cor:threshold-from-stats-chernoff}, it suffices to show that
\begin{equation}
    \label{eq:error-bound}
     \Ex_{\bell \sim \mu^k(T)}\bracket*{g_{\delta k/10}(\Dens_H(\bell) - \Adv_H(\bell))}  \leq e^{-\delta k/10} + 1-\gamma.
\end{equation}
Recall that we have much information about how $\Dens_H(\bell)$ distributes over the leaves via \Cref{lem:resilience-general}, but a coarser understanding of how $\Adv_H(\bell)$ distributes via \Cref{lem:expected total hardcore advantage body}. Because of this, we will first bound the above equation where the $\Adv_H(\bell)$ is set to $0$ and analyze how much including that term affects the result.
\begin{lemma}
    \label{lem:error-no-advantange}
    For any tree $T:(\bn)^k \to \bits^k$ and hardcore measure of density $\delta$ w.r.t. distribution $\mu$,
    \begin{equation*}
        \Ex_{\bell \sim \mu^k(T)}\bracket*{g_{\delta k/10}(\Dens_H(\bell))} \leq e^{-.121\delta k}.
    \end{equation*}
\end{lemma}
\begin{proof}
    We bound,
    \begin{equation*}
        \Ex_{\bell \sim \mu^k(T)}\bracket*{\min\paren[\Big]{1, \exp\paren[\Big]{\delta k / 10 - \frac{\Dens_H(\bell)}{4}}}} \leq e^{\delta k/10} \cdot \Ex_{\bell \sim \mu^k(T)}\bracket*{e^{-\Dens_H(\ell)/4}}.
    \end{equation*}
    Since $z \mapsto e^{-z/4}$ is convex, we can use \Cref{lem:resilience-general} to bound the above using the moment generating function of the binomial distribution,
    \begin{align*}
         \Ex_{\bell \sim \mu^k(T)}\bracket*{e^{-\Dens_H(\bell)/4}} &\leq \Ex_{\bz \sim \Bin(k,\delta)}[e^{-\bz/4}] \\
         &= (1 - \delta(1 - e^{-1/4}))^k \\
         & \leq e^{-(1 -e^{-1/4})\delta k}.
    \end{align*}
    Combining the above,
    \begin{equation*}
         \Ex_{\bell \sim \mu^k(T)}\bracket*{g_{\delta k/10}(\Dens_H(\bell))}  \leq e^{-\delta k (1 - e^{-1/4} - 1/10)} \leq e^{-0.121 \delta k}. \qedhere
    \end{equation*}
\end{proof}
Next, we prove a Lipschitz-style bound for $g$. This will be useful in incorporating $\Adv_H(\bell)$ to \Cref{eq:error-bound}.
\begin{proposition}
    \label{prop:lipschitz-bounds}
    For any $z, \Delta, t \geq 0$,
    \begin{equation}
        \label{eq:lipschitz}
        g_t(z - \Delta) \leq g_t(z) + \Delta/4.
    \end{equation}
    Furthermore, if $z \geq 5t$, then
    \begin{equation}
        \label{eq:lipschitz-scaled}
        g_t(z - \Delta) \leq g_t(z) + \Delta/t.
    \end{equation}
\end{proposition}
\begin{proof}
    \Cref{eq:lipschitz} follows from the $(1/4)$-Lipschitzness of $g_t(z)$.

    For \Cref{eq:lipschitz-scaled}, fix any choice of $z \geq 5t$. We want to show that for any choice of $\Delta$,
    \begin{equation*}
        \frac{ g_t(z - \Delta) - g_t(z)}{\Delta} \leq 1/t.
    \end{equation*}
    We claim that the left hand side of the above inequality is maximized when $\Delta = z - 4t$. When $\Delta$ is increased beyond $z - 4t$, the numerator remains constant (because $g_t(z)$ is constant for any $z \leq 4t$, but the denominator increases, so the maximum cannot occur at at any $\Delta > z-4t$. On the other hand, $g_t(z)$ is convex when restricted to the domain $[4t, \infty)$, so the maximum cannot occur at any $\Delta < z - 4t$. Therefore, it suffices to consider $\Delta = z - 4t$, in which case,
    \begin{equation*}
         \frac{ g_t(z - \Delta) - g_t(z)}{\Delta} = \frac{1 - g_t(z)}{\Delta} \leq \frac{1}{\Delta} \leq 1/t. \qedhere
    \end{equation*}
\end{proof}

We are now ready to prove the main result of this section.
\begin{proof}[Proof of \Cref{thm:bounds-from-hardcore}]
    By applying \Cref{cor:threshold-from-stats-chernoff},
    \begin{equation*}
        \Prx_{\bX \sim \mu^k}\bracket*{\lzero[\big]{T(\bX)-f^{\otimes k}(\bX)} \leq \delta k/10} \leq \Ex_{\bell \sim \mu^k(T)}\bracket*{g_{\delta k/10}(\Dens_H(\bell) - \Adv_H(\bell))}.
    \end{equation*}
    First, we consider the case where $\delta k \leq 40$. Here, by applying \Cref{eq:lipschitz},
    \begin{align*}
        \Ex_{\bell \sim \mu^k(T)}\bracket*{g_{\delta k/10}(\Dens_H(\bell) - \Adv_H(\bell))} &\leq \Ex_{\bell \sim \mu^k(T)}\bracket*{g_{\delta k/10}(\Dens_H(\bell))} + \tfrac{1}{4}\cdot \Ex_{\bell \sim \mu^k(T)}[\Adv_H(\bell)]\\
        &\leq e^{-0.121 \delta k} + \gamma \delta k/4 \tag{\Cref{lem:error-no-advantange,lem:expected total hardcore advantage body}}\\
        &\leq e^{-\delta k/10} + 10\gamma\tag{$\delta k \leq 40$}
    \end{align*}

    When $\delta k > 40$, we break down the desired result into two pieces, depending on whether $\Dens_H(\bell)$ is small or large. For the piece where $\Dens_H(\bell)$ is small, we just use that $g(\cdot)$ is bounded between $0$ and $1$ which means $g(z) - g(z-\Delta) \leq 1$,
    \begin{align*}
         \Ex\big[g_{\delta k/10}(\Dens_H(\bell) &- \Adv_H(\bell)) \cdot \Ind[\Dens_H(\bell) \leq \delta k/2]\big] \\
         &\leq \Ex\bracket*{g_{\delta k/10}(\Dens_H(\bell))\cdot \Ind[\Dens_H(\bell) \leq \delta k/2]} + \Pr[\Dens_H(\bell) \leq \delta k/2]\\
         &\leq  \Ex\bracket*{g_{\delta k/10}(\Dens_H(\bell))\cdot \Ind[\Dens_H(\bell) \leq \delta k/2]}  + e^{-\delta k/8}. \tag{\Cref{lem:resilience lemma}}
    \end{align*}
    For the piece where $\Dens_H(\bell)$ is large, we use \Cref{eq:lipschitz-scaled}
    \begin{align*}
        \Ex\big[g_{\delta k/10}&(\Dens_H(\bell) - \Adv_H(\bell)) \cdot \Ind[\Dens_H(\bell) > \delta k/2]\big] \\
        &\leq \Ex\bracket*{g_{\delta k/10}(\Dens_H(\bell))\cdot \Ind[\Dens_H(\bell) > \delta k/2]} + \tfrac{10}{\delta k} \cdot \Ex\bracket*{\Adv_H(\bell)\cdot \Ind[\Dens_H(\bell) > \delta k/2]} \tag{\Cref{eq:lipschitz-scaled}} \\
        &\leq \Ex\bracket*{g_{\delta k/10}(\Dens_H(\bell))\cdot \Ind[\Dens_H(\bell) > \delta k/2]} + \tfrac{10}{\delta k} \cdot \Ex\bracket*{\Adv_H(\bell)} \tag{$\Dens_H(\bell) \geq 0$} \\
        &\leq \Ex\bracket*{g_{\delta k/10}(\Dens_H(\bell))\cdot \Ind[\Dens_H(\bell) > \delta k/2]} + \tfrac{10}{\delta k} \cdot \gamma \delta k \tag{\Cref{lem:expected total hardcore advantage body}} 
    \end{align*}
    Combining the above two pieces,
    \begin{align*}
        \Ex\big[g_{\delta k/10}(\Dens_H(\bell) - \Adv_H(\bell))\big] &\leq \Ex\big[g_{\delta k/10}(\Dens_H(\bell))\big] + e^{-\delta k/8} + 10\gamma \\
        &\leq e^{-\delta k/8} + e^{-0.121 \delta k} + 10\gamma \tag{\Cref{lem:error-no-advantange}}
    \end{align*}
    When $\delta k > 40$, $e^{-\delta k/8} + e^{-0.121 \delta k} < e^{-\delta k/10}$, so we also recover the desired result in this case.
\end{proof}

\section{Equivalence between direct sum theorems and XOR lemmas and the proof of~\Cref{thm:main-xor}}

In this section, we prove the following claim which shows that a strong direct sum theorem implies a strong XOR lemma. We then derive~\Cref{thm:main-xor} as a consequence of this equivalence and our strong direct sum theorem for query complexity (\Cref{thm:main-intro}).

\begin{claim}[Equivalence between direct sum theorems and XOR lemmas]
\label{claim:equivalence between direct product and XOR}
For every $f:\bn\to\bits$, distribution $\mu$ over $\bn$, integer $k\in\N$, multiplicative factor $M \in \R$, and $\eps\in (0,1)$, if the following direct sum theorem holds:
$$
\overline{\mathrm{Depth}}^{\mu^k}(f^{\otimes k},\eps)\ge M\cdot\overline{\mathrm{Depth}}^{\mu}(f,\delta).
$$
then, the following XOR lemma holds:
$$
\overline{\mathrm{Depth}}^{\mu^k}(f^{\oplus k},\tfrac{\eps}{2})\ge M\cdot\overline{\mathrm{Depth}}^{\mu}(f,\delta).
$$
\end{claim}

In order to prove \Cref{claim:equivalence between direct product and XOR} and \Cref{thm:main-xor}, we establish a lemma which allows us to convert any decision accurately computing $f^{\oplus k}$ into a decision tree accurately computing $f^{\otimes k}$. The following definition captures this conversion.

\begin{definition}[The product tree]
\label{def:product tree}
    Given a decision tree $T:(\bn)^k \to \bits$, the $k$-wise product tree $\Tilde{T}:(\bn)^k \to \bits^k$ is defined as follows. For the internal nodes, $\Tilde{T}$ has exactly the same structure as $T$. For a leaf $\ell$ in $T$, the leaf vector $(\ell_1,\ldots,\ell_k)\in\bits^k$ in $\Tilde{T}$ is defined by
    $$
    \ell_i\coloneqq\sign\left(\Ex_{\bX\sim\mu^{k}}[f(\bX^{(i)})\mid \bX\text{ reaches }\ell]\right)
    $$
    for all $i\in[k]$.
\end{definition}

Intuitively, $\Tilde{T}$ computes $T$'s best guess for $f(X^{(i)})$ for each $i\in [k]$ on a given input $(X^{(1)},\ldots,X^{(k)})$. If $T$ is really good at computing $f^{\oplus k}$ then at every leaf it should have queried enough variables to pin down $f$'s value on each of the input blocks. The main lemma formalizes this intuition.

\begin{lemma}
\label{lem:dpt-to-ds}
    For any $f:\bits^n \to \bits$, distribution $\mu$ over $\bits^n$, and tree $T:(\bn)^k \to \bits$, the $k$-wise product tree $\Tilde{T}:(\bn)^k \to \bits^k$ satisfies
    \begin{equation*}
        \Prx_{\bX \sim \mu^{k}}[\tilde{T}(\bX) = f^{\otimes k}(\bX)] \geq \Ex_{\bX \sim \mu^{k}}[T(\bX)f^{\oplus k}(\bX)].
    \end{equation*}
\end{lemma}

\subsection{Proofs of~\Cref{claim:equivalence between direct product and XOR} and~\Cref{thm:main-xor} assuming~\Cref{lem:dpt-to-ds}}
The following corollary of \Cref{lem:dpt-to-ds} implies~\Cref{claim:equivalence between direct product and XOR} and~\Cref{thm:main-xor}.
\begin{corollary}[Main corollary of \Cref{lem:dpt-to-ds}]
    \label{cor:xor-to-dp-for-r}
    For all $f:\bits^n\to\bits$, $k\ge 1$, distributions $\mu$ over $\bits^n$, and $\eps>0$, $\overline{\mathrm{Depth}}^{\mu^{k}}(f^{\oplus k},\frac{\eps}{2})\ge \overline{\mathrm{Depth}}^{\mu^{k}}(f^{\otimes k},\eps)$.
\end{corollary}
\begin{proof}
    Let $\bA$ be a randomized query algorithm for $f^{\oplus k}$ with error $\eps/2$ and expected cost $q=\overline{\mathrm{Depth}}^{\mu^{k}}(f^{\oplus k},\eps/2)$. Let $\mathcal{T}$ denote the distribution over decision trees determined by $\bA$. Consider the algorithm $\Tilde{\bA}$ which computes $f^{\otimes k}(X)$ by sampling $\bT\sim\mathcal{T}$ and returning $\Tilde{\bT}(X)$ where $\Tilde{\bT}$ is the decision tree from \Cref{lem:dpt-to-ds}. Then, the success of $\Tilde{\bA}$ is
    \begin{align*}
        \Ex_{\bT\sim \mathcal{T}}\left[\Prx_{\bX\sim \mu^{k}}[\Tilde{\bT}(\bX)=f^{\otimes k}(\bX)]\right]&\ge \Ex_{\bT\sim \mathcal{T}}\left[\Ex_{\bX\sim\mu^{k}}[\bT(\bX)f^{\oplus k}(\bX)]\right]\tag{\Cref{lem:dpt-to-ds}}\\
        &\ge 1-\eps
    \end{align*}
    where the last step uses the fact that advantage is $1-2\cdot\mathrm{error}$. Since the structure of each $\Tilde{T}$ is the same as $T$, the expected cost of $\Tilde{\bA}$ is $q$ which completes the proof.
\end{proof}

\begin{proof}[Proofs of~\Cref{claim:equivalence between direct product and XOR} and~\Cref{thm:main-xor}]
    By \Cref{cor:xor-to-dp-for-r}, $$\overline{\mathrm{Depth}}^{\mu^{k}}(f^{\oplus k},\tfrac{\eps}{2})\ge \overline{\mathrm{Depth}}^{\mu^{k}}(f^{\otimes k},\eps)\ge M\cdot\overline{\mathrm{Depth}}^{\mu}(f,\delta).$$
    \Cref{thm:main-xor} follows immediately by applying \Cref{claim:equivalence between direct product and XOR} to \Cref{thm:main-formal}.
\end{proof}

\begin{remark}[On the necessity of the $1/2$ loss in $\eps$ in \Cref{cor:xor-to-dp-for-r}]
    \label{rem:loss}
    One may wonder whether the $1/2$ loss in $\eps$ parameter in \Cref{cor:xor-to-dp-for-r} is necessary. For example, can one show $\overline{\mathrm{Depth}}^{\mu^{k}}(f^{\oplus k},0.51\eps)\ge \overline{\mathrm{Depth}}^{\mu^{k}}(f^{\otimes k},\eps)$? The issue is that $\overline{\mathrm{Depth}}^{\mu^{k}}(f^{\oplus k},0.5)=0$ for all functions $f:\bn\to\bits$ because the bias of $f^{\oplus k}$ is at least $0.5$. So such a statement cannot hold for all $f$ in all parameter regimes. Concretely, one can show that if $f$ is the parity of $n$ bits and $\mu$ is uniform over $\bn$, then $\overline{\mathrm{Depth}}^{\mu^{k}}(f^{\otimes k},\eps)\ge \Omega(kn)$ for all constant $\eps<1$. Any path in a decision tree for $f^{\otimes k}$ which queries at most $\lambda kn$ bits for some constant $\lambda<1$ has success probability $2^{-\Omega(k)}$. So to achieve any constant accuracy requires $\Omega(kn)$ expected depth. On the other hand, $\overline{\mathrm{Depth}}^{\mu^{k}}(f^{\oplus k},0.5)=1\ll \overline{\mathrm{Depth}}^{\mu^{k}}(f^{\otimes k},\eps)$ which is achieved by the decision tree that outputs a single constant value. Therefore, the $\eps/2$ in \Cref{cor:xor-to-dp-for-r} is necessary for such a statement to hold in full generality.
\end{remark}

\subsection{Proof of~\Cref{lem:dpt-to-ds}}

Each $\ell_i$ for $i\in [k]$ satisfies
\begin{align*}
    \Ex_{\bX\sim\mu^{k}}[\ell_i\cdot f(\bX^{(i)})\mid\bX\text{ reaches }\ell]&=\left|\Ex_{\bX\sim\mu^{k}}[f(\bX^{(i)})\mid\bX\text{ reaches }\ell]\right|\\
    &\ge \Ex_{\bX\sim\mu^{k}}[\ell\cdot f(\bX^{(i)})\mid\bX\text{ reaches }\ell].
\end{align*}
In particular, for all {leaves} $\ell$ of $T$,
\begin{equation}
\label{eq:corr}
    \Ex_{\bX\sim\mu^{k}}[\Tilde{T}(\bX)_i\cdot f(\bX^{(i)})\mid\bX\text{ reaches }\ell]\ge \Ex_{\bX\sim\mu^{k}}[T(\bX)\cdot f(\bX^{(i)})\mid\bX\text{ reaches }\ell].
\end{equation}

Therefore:
    \begin{align*}
    \Prx_{\bX \sim \mu^{k}}&[\tilde{T}(\bX)= f^{\otimes k}(\bX)]\\
        &= \Ex_{\bell\sim\mu^{k}(T)}\bigg[\Prx_{\bX\sim\mu^{k}}[\tilde{T}(\bX)=f^{\otimes k}(\bX)\mid \bX\text{ reaches }\bell]\bigg]\\
        &=\Ex_{\bell\sim\mu^{k}(T)}\bigg[\prod_{i\in [k]}\Prx_{\bX\sim\mu^{k}}[\tilde{T}(\bX)_i=f(\bX^{(i)})\mid \bX\text{ reaches }\bell]\bigg]\tag{\Cref{claim:leaf-product}}\\
        &\ge \Ex_{\bell\sim\mu^{k}(T)}\bigg[\prod_{i\in [k]}\Ex_{\bX\sim\mu^{k}}[\Tilde{T}(\bX)_if(\bX^{(i)})\mid \bX\text{ reaches }\bell]\bigg]\\
        &\ge \Ex_{\bell\sim\mu^{k}(T)}\bigg[\prod_{i\in [k]}\Ex_{\bX\sim\mu^{k}}[T(\bX)f(\bX^{(i)})\mid \bX\text{ reaches }\bell]\bigg]\tag{\Cref{eq:corr}}\\
        &=\Ex_{\bell\sim\mu^{k}(T)}\bigg[\Ex_{\bX\sim\mu^{k}}\bigg[T(\bX)\prod_{i\in [k]}f(\bX^{(i)})\mid \bX\text{ reaches }\bell\bigg]\bigg]\tag{\Cref{claim:leaf-product}}\\
        &=\Ex_{\bell\sim\mu^{k}(T)}\bigg[\Ex_{\bX\sim\mu^{k}}[T(\bX)f^{\otimes k}(\bX)\mid \bX\text{ reaches }\bell]\bigg]\tag{$\prod_{i\in [k]}f(\bX^{(i)})=f^{\oplus k}(\bX)$}\\
        &=\Ex_{\bX \sim \mu^{k}}[T(\bX)f^{\oplus k}(\bX)]
    \end{align*}
    which completes the proof.\hfill\qed




\section{Proof of~\Cref{claim:linear dependence on gamma}}


\begin{claim} [The $\gamma$ factor in \Cref{thm:main-formal} is necessary; \Cref{claim:linear dependence on gamma} restated]
    \label{claim:linear-dependence}
    Let $\mathsf{Par} : \bn\to\bits$ be the parity function and $\mu $ be the uniform distribution over $\bn$. Then for all $\gamma$,
    \[
\overline{\mathrm{Depth}}^{\mu^k}(\mathsf{Par}^{\otimes k}, 1 -\gamma) \leq2 \gamma k \cdot \overline{\mathrm{Depth}}^\mu(\mathsf{Par}, \tfrac1{4}).
    \]
\end{claim}

We will need the following simple proposition, which states that in any tree that seeks to compute the $n$-variable parity function, leaves of depth strictly less than $n$ contribute $\frac1{2}$ to the error:
\begin{proposition}
    \label{prop:par-query-all}
    For any (potentially randomized) tree $T: \bn \to \bits$ and leaf $\ell$ of $T$ with depth strictly less than $n$,
    \begin{equation*}
        \Prx_{\bx \sim \mathrm{Unif}(\bn)}[T(\bx) \neq \mathsf{Par}(\bx) \mid \bx \text{ reaches } \ell] = \tfrac{1}{2}.
    \end{equation*}
\end{proposition}
\begin{proof}
    Since $\ell$ is at depth strictly less than $n$, there must be some index $i \in [n]$ not queried on the path to $\ell$. Taking any input $x$ that reaches $\ell$, the input $x'$ with the $i^{\text{th}}$ bit flipped must also reach $\ell$ and have the opposite parity. Both of these inputs are equally likely under the uniform distribution and so the value of $\mathsf{Par}(\bx)$ conditioned on $\bx$ reaching $\bell$ is equally likely to be $+1$ and $-1$. Therefore, $T$ errs half the time it reaches this leaf regardless of how it labels it.
\end{proof}

\begin{proof}[Proof of \Cref{claim:linear-dependence}]
The proof proceeds in two parts. First, we show that $\overline{\text{Depth}}^\mu(\mathsf{Par}, \tfrac{1}{4})= \frac{n}{2}$. Second, we prove that  $\overline{\text{Depth}}^{\mu^{\otimes k}}(\mathsf{Par}^{\otimes k}, 1 -\gamma) \leq \gamma k n$.

\pparagraph{(1) $\overline{\mathrm{Depth}}^\mu(\mathsf{Par}, \tfrac{1}{4})= \tfrac{n}{2}$.} Let $T$ be an arbitrary randomized decision tree let $p_n(T)$ be the probability that $T$ queries all $n$ variables. Then, the expected depth of $T$ is at least $n \cdot p_n(T)$. Meanwhile, by \Cref{prop:par-query-all}, the error of $T$ in computing parity is at least $\frac1{2} \cdot p_n(T)$ w.r.t. the uniform distribution. Therefore, $\overline{\text{Depth}}^\mu(f, \frac1{4}) \geq \frac{n}{2}$.

While this direction is not needed for \Cref{claim:linear-dependence} we show for completeness that $\overline{\text{Depth}}^\mu(f, \frac1{4}) \leq \frac{n}{2}$ by constructing a randomized\footnote{At the cost of increasing expected depth by $1$, the tree can be derandomized. To derandomize it, read a single bit of the input and only query the rest if that bit is $1$, which occurs with probability $\tfrac{1}{2}$.} decision tree $T$ for $f$.
 With probability $\frac1{2}$, $T$ queries all $n$ variables to compute $f$ exactly. Otherwise, it simply outputs $0$. $T$ has expected depth $\frac{n}{2}$, and it errs only when it queries no variables and guesses incorrectly, which happens with probability $\frac1{2}\cdot \frac1{2}= \frac1{4}$. Thus, $\overline{\text{Depth}}^\mu(f, \frac1{4}) \leq \frac{n}{2}$. 



\pparagraph{(2)  $\overline{\mathrm{Depth}}^{\mu^{\otimes k}}(f^{\otimes k}, 1 -\gamma) \leq  \gamma k n$.} We construct a randomized\footnote{Again, this can be derandomized at the cost of adding 
$\leq 2$ to the expected depth, since any biased coin can be simulated with a fair coin using $2$ flips in expectation.} decision tree $T$ for $f^{\otimes k}$. With probability $\gamma$, $T$ queries all $kn$ variables to compute $f^{\otimes k}$ exactly, and with probability $1-\gamma$, it outputs 0. When it queries all variables, it has no error so its average error is at most $1 - \gamma$. Furthermore, its average depth is $\gamma k n$.
\end{proof}

\section*{Acknowledgements}

We thank William Hoza for a helpful conversation and the CCC reviewers for their comments and feedback. 

The authors are supported by NSF awards 1942123, 2211237, 2224246, a Sloan Research Fellowship, and a Google Research Scholar Award. Caleb is also supported by an NDSEG Fellowship, and Carmen by a Stanford Computer Science Distinguished Fellowship and an NSF Graduate Research Fellowship.

\bibliographystyle{alpha}
\bibliography{ref}

\newpage 
\appendix
\section{Figures of stacked and fair trees}
\label{app:figures}

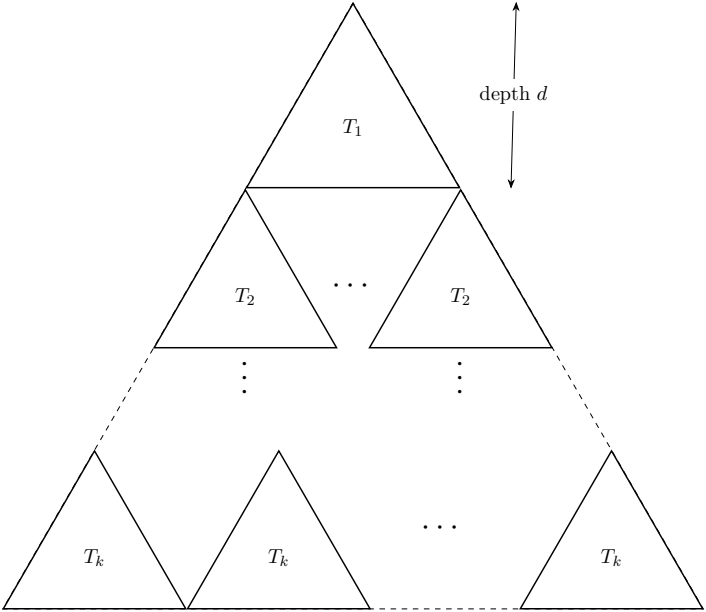
\begin{figure}[h!]
    \centering
    \scalebox{0.7}{
    \begin{tikzpicture}[tips=proper]
        \node[isosceles triangle,
            draw,thick,
            isosceles triangle apex angle=60,
            rotate=90,
            minimum size=3.5cm] (T1) at (0,0){};
        
        \node[isosceles triangle,
            draw,thick,
            isosceles triangle apex angle=60,
            rotate=90,
            minimum size=3cm,
            anchor=east] (T2L) at (T1.left corner){};

        \node[isosceles triangle,
            draw,thick,
            isosceles triangle apex angle=60,
            rotate=90,
            minimum size=3cm,
            anchor=east] (T2R) at (T1.right corner){};
    
        \node[isosceles triangle,
            draw,
            dashed,
            isosceles triangle apex angle=60,
            rotate=90,
            minimum size=11.52cm,
            anchor=east] (T5) at (T1.east){};
        
        \node[isosceles triangle,
            draw,
            thick,
            isosceles triangle apex angle=60,
            rotate=90,
            minimum size=3cm,
            anchor=east] (T3LL) at ([yshift=-0.00cm]T5.99){};
        \node[isosceles triangle,
            draw,
            thick,
            isosceles triangle apex angle=60,
            rotate=90,
            minimum size=3cm,
            anchor=east] (T3L) at ([xshift=3.5cm,yshift=-0.00cm]T5.99){};
        \node[isosceles triangle,
            draw,thick,
            isosceles triangle apex angle=60,
            rotate=90,
            minimum size=3cm,
            anchor=east] (T3R) at ([yshift=-0.00cm]T5.261){};    
        
        \draw[] (0,-3) node [] {\LARGE $\cdots$};
        \draw[] (1.7,-7.6) node [] {\LARGE $\cdots$};
        \draw[] ([yshift=-0.5cm]T2L.west) node [rotate=90] {\LARGE $\cdots$};
        \draw[] ([yshift=-0.5cm]T2R.west) node [rotate=90] {\LARGE $\cdots$};

        \draw[] (T1.center) node [] {$T_1$};
        \draw[] (T2L.center) node [] {$T_2$};
        \draw[] (T2R.center) node [] {$T_2$};
        \draw[] (T3L.center) node [] {$T_k$};
        \draw[] (T3LL.center) node [] {$T_k$};
        \draw[] (T3R.center) node [] {$T_k$};

        \draw[{Stealth[scale=1]}-{Stealth[scale=1]}] ([xshift=3.1cm]T1.east) to node[midway,fill=white!30,scale=1] {depth $d$} ([xshift=3cm]T1.west);
               
    \end{tikzpicture}
    }
  \captionsetup{width=.9\linewidth}

    \caption{Illustration of a stacked decision tree for a function $f^{\otimes k}$. The decision tree consists of $k$ depth-$d$ decision trees, $T_1,\ldots, T_k$, stacked on top of each other. For an input $X \in (\bn)^k$, the output $T(X)$ is computed sequentially, first by computing $T_1(X)$, then $T_2(X)$, and so on. The final output is $T(X)\coloneqq (T_1(X),\ldots, T_k(X))$. }
    \label{fig:stacked}
\end{figure}

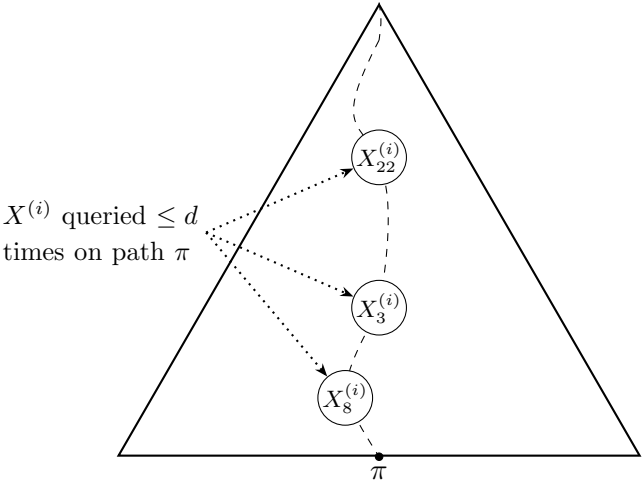
\begin{figure}[h!]
    \centering
    \begin{tikzpicture}[tips=proper]
        \node[isosceles triangle,
            draw,
            thick,
            isosceles triangle apex angle=60,
            rotate=90,
            minimum size=6cm] (T1) at (0,0){};
        
        \draw[black,dashed] (T1.east) .. controls ([xshift=-0.1cm]T1.358) .. ([yshift=-0.5cm]T1.east) node[] (N1) {{}};
        \draw[black,dashed] ([yshift=-0.5cm]T1.east) .. controls ([xshift=0.4cm]T1.40) and (T1.350) .. (T1.center) node[] (N1) {{}};
        \draw[black,dashed] (T1.center) .. controls ([xshift=2.4cm]T1.110) .. (T1.west) node[] (N2) {{}};

        \node[draw,circle,inner sep=0pt,fill=white] (x3) at (-0.45,-1.2) {\footnotesize $X_{8}^{(i)}$};
        \node[draw,circle,inner sep=0pt,fill=white] (x2) at (0,0) {\footnotesize $X_{3}^{(i)}$};
        \node[draw,circle,inner sep=0pt,fill=white] (x1) at (0,2) {\footnotesize $X_{22}^{(i)}$};
        
        \draw[color=black] (T1.west) node [below,fill=white] {$\pi$};
        \node[draw,circle,fill=black,inner sep=1pt] (x) at (T1.west) {};

        \draw[] (-3.7,1) node [fill=white,text width=26mm] {\small $X^{(i)}$ queried $\le d$ times on path $\pi$};
        \coordinate (p) at (-2.3,1);
        \draw[-{Stealth[scale=0.75]},dotted,thick] (p) to (x1);
        \draw[-{Stealth[scale=0.75]},dotted,thick] (p) to (x2);
        \draw[-{Stealth[scale=0.75]},dotted,thick] (p) to (x3);

    \end{tikzpicture}
  \captionsetup{width=.9\linewidth}
    \caption{Illustration of a fair decision tree. For every block $i\in [k]$ and path $\pi$, the input block $X^{(i)}$ is queried at most $d$ times.} 
    \label{fig:fair}
\end{figure}

\section{Proof of~\Cref{thm:hardcore theorem for DTs}}
\label{appendix:hardcore}
Let $\mathcal{H}$ denote the set of measures of density $\delta/2$ with respect to $\mu$ and let $\mathcal{T}$ denote the set of decision trees $T$ whose expected depth with respect to $\mu$ is at most $d$. Suppose for contradiction that there \textit{does not} exist an $H\in\mathcal{H}$ which is $(\gamma,d)$-hardcore. That is, for all $H\in\mathcal{H}$ there is a tree $T$ of expected depth at most $d$ satisfying
\begin{equation}
\label{eq:minimax}
\Ex_{\bx\sim\mu}[f(\bx)T(\bx)H(\bx)]>\gamma\Ex_{\bx\sim\mu}[H(\bx)]=\gamma\delta/2.
\end{equation}
We use the minimax theorem to switch the quantifiers in the above statement. Consider the payoff matrix $M$ whose rows are indexed by distributions from $\mathcal{H}$ and whose columns are indexed by algorithms from $\mathcal{T}$ and whose entries are given by $M_{H,T}\coloneqq \Ex_{\bx\sim\mu}[f(\bx)T(\bx)H(\bx)]$. This is the payoff matrix for the zero-sum game where the row player first chooses a row $H$ and the column player then chooses a column $T$ and the payoffs are determined by $M_{H,T}$. Note that once the first player's strategy is fixed, we can assume without loss of generality that the second player's strategy is deterministic. Therefore, the minimax theorem for zero-sum games yields
\begin{align*}
    \gamma\delta/2&<\min_{\rho\in\mu(\mathcal{H})}\max_{T\in\mathcal{T}}~ (\rho^{\top} M)_T\tag{\Cref{eq:minimax}}\\
    &=\max_{\tau\in\mu(\mathcal{T})}\min_{H\in\mathcal{H}}~(M\tau)_H\tag{minimax theorem}
\end{align*}
where $\mu(\cdot)$ denotes the set of distributions over a given set. Therefore, there is a fixed distribution $\tau$ over the set $\mathcal{T}$ such that for all $H\in\mathcal{H}$
\begin{equation}
\label{eq:hardcore}
\Ex_{\bT\sim\tau}\bigg[\Ex_{\bx\sim\mu}[f(\bx)\bT(\bx)H(\bx)]\bigg]>\gamma\delta/2.    
\end{equation}
This shows that
$$
\Prx_{\bx\sim\mu}\bigg[\Ex_{\bT\sim\tau}[\bT(\bx)]f(\bx)\ge \gamma\bigg]\ge 1-\delta/2.
$$
In particular, if instead $\Prx_{\bx\sim\mu}[\Ex_{\bT\sim\tau}[\bT(\bx)]f(\bx)<\gamma]\ge \delta/2$ then we can contradict \Cref{eq:hardcore} by constructing a $\delta/2$-density $H$ such that $H(x)\coloneqq\Pr_{\bx\sim\mu}[\bx=x]$ for a $\delta/2$-fraction of $x$ satisfying $\Ex_{\bT\sim\tau}[\bT(\bx)]f(\bx)< \gamma$. \Cref{eq:hardcore} shows that $\Ex_{\bT\sim\tau}[\bT(\bx)]$ has good correlation with $f$ for a large fraction of inputs. We obtain a single strategy from the distribution $\tau$ by sampling $\bT_1,\ldots,\bT_r\sim\tau$ for $r$ sufficiently large (chosen later) and defining $T^\star$ as $\bT^{\star}(x)\coloneqq \text{MAJ}(\bT_1(x),\ldots,\bT_r(x))$.
For every $x$ for which $\Ex_{\bT\sim\tau}[\bT(x)]f(x)\ge \gamma$, we have
$$
\Prx_{\bT_1,\ldots,\bT_r\sim\tau}\bigg[\text{MAJ}(\bT_1(x),\ldots,\bT_r(x))\neq f(x)\bigg]\le 2^{-\Omega(\gamma^2 r)}
$$
by a Chernoff bound. Choosing $r=\Theta(\log(1/\delta)/\gamma^2)$ ensures that the failure probability is at most $\delta/2$. The decision tree $\bT^{\star}$ satisfies $\Pr_{\bx\sim\mu}[\bT^\star(\bx)\neq f(\bx)]\le \delta$. The expected depth of $\bT^\star$ is less than
$$
r\cdot d=\Theta(d\log (1/\delta)/\gamma^2)<\overline{\mathrm{Depth}}^\mu(f,\delta)
$$
which is a contradiction.

\section{The lack of error reduction for distributional error}
\label{sec:no-boosting}
In \Cref{sec:challenges}, we showed how error reduction gave a simple proof of a strong direct sum theorem for \emph{randomized} query complexity. The specific statement needed in that proof is the following standard error reduction by repetition theorem.
\begin{fact}[Error reduction for $\overline{\mathrm{R}}$]
    For any function $f:\bits^n \to \bits$ and $\delta > 0$,
    \begin{equation*}
        \overline{\mathrm{R}}(f,\delta) \leq O(\log(\lfrac{1}{\delta}))\cdot \overline{\mathrm{R}}(f,1/4).
    \end{equation*}
\end{fact}
Here, we give a short proof that no error reduction holds in the distributional setting, even with substantially weaker parameters.
\begin{claim}
    \label{claim:no-boosting}
    For any $n \in \N$, let $\mu$ be the uniform distribution over $\bits^n$. There is a function $f:\bits^n \to \bits$ satisfying,
    \begin{equation*}
        \overline{\mathrm{Depth}}^{\mu}(f,1/4) = 0 \quad\quad\quad\text{and}\quad\quad\quad \overline{\mathrm{Depth}}^{\mu}(f,1/8) \geq \Omega(n).
    \end{equation*}
\end{claim}
Since any function on $n$ bits can be computed exactly using $n$, the $f$ in the above claim requires essentially the maximum number of queries to be computed to error $1/8$ despite requiring no queries to be computed to error $1/4$.
\begin{proof}
    We first define $f$: If $x_1 = 0$, then $f(x) = 0$. Otherwise, $f(x)$ is the parity of the remaining $n-1$ bits of $x$.

    Note that,
    \begin{equation*}
        \Pr_{\bx \sim \mu}[f(\bx) = 0] = \frac{1}{2} \cdot \paren*{\Pr[f(\bx) = 0 \mid \bx_1 = 0] + \Pr[f(\bx) = 0 \mid \bx_1 = 1]} = \frac{3}{4}.
    \end{equation*}
    Therefore, the $0$ query algorithm that simply outputs $0$ has an error of only $1/4$ on $f$. 
    
    It only remains to prove that $\overline{\mathrm{Depth}}^{\mu}(f,1/10) \geq \Omega(n)$. Consider any (potentially randomized) $T:\bits^n \to \bits$ and leaf $\ell$ of $T$ at depth strictly less than $n-1$. By \Cref{prop:par-query-all}
    \begin{equation*}
        \Pr_{\bx \sim \mu}[T(\bx) \neq f(\bx) \mid \bx \text{ reaches } \ell, \bx_1 = 1] = 1/2.
    \end{equation*}
    Let $p$ be the probability that $T(\bx)$ queries a leaf of depth at least $n-1$ given that $\bx_1 = 1$. The above allows us to conclude that
    \begin{equation*}
        \Pr_{\bx \sim \mu}[T(\bx) \neq f(\bx)] \geq \frac{1}{2} \cdot \Pr[T(\bx) \neq f(\bx)\mid \bx_1 = 1] \geq \frac{1}{4} \cdot p. 
    \end{equation*}
    Therefore, if $T$ has error at most $1/8$, then $p$ must be at least $1/2$, which shows that $\overline{\mathrm{Depth}}^{\mu}(f,1/8) \geq \frac{n-1}{4}$.
\end{proof}

\end{document}